\documentclass{article}

\usepackage[final,nonatbib]{neurips_2025}

\usepackage[utf8]{inputenc} 
\usepackage[T1]{fontenc}    
\usepackage{url}            
\usepackage{booktabs}       
\usepackage{amsfonts}       
\usepackage{nicefrac}       
\usepackage{microtype}      
\usepackage{xcolor}         

\usepackage{times}
\usepackage{tikz}
\usepackage{mathtools}
\usepackage{amsthm}
\usepackage{amssymb}
\usepackage{bbm}
\usepackage{verbatim}
\usepackage{enumitem}
\usepackage{xspace}
\usepackage{todonotes}
\usepackage{amsfonts,mathrsfs,color}
\usepackage{xcolor}
\usepackage{graphicx}
\usepackage{booktabs}
\usepackage{nicefrac}
\usepackage{mdframed}
\usepackage{contour}
\contourlength{0.1pt}
\contournumber{10}
\usepackage[
    backend=biber,
    style=alphabetic,
    maxcitenames=3,
    maxbibnames=99,
    natbib=true,
    url=false,
    doi=false, 
    isbn=false]{biblatex}
\usepackage{multirow}

\usepackage[ruled,linesnumbered,noend]{algorithm2e} 

\SetCommentSty{mycommfont}
\SetKwComment{Comment}{$\triangleright$~}{}
\SetKwProg{Fn}{function}
\SetAlFnt{\relax}
\SetAlCapFnt{\relax}
\SetAlCapNameFnt{\relax}
\SetAlCapHSkip{0pt}
\usepackage{tikz}
\usepackage{xcolor}

\newcommand{\declarecolor}[2]{\definecolor{#1}{RGB}{#2}\expandafter\newcommand\csname #1\endcsname[1]{\textcolor{#1}{##1}}}
\declarecolor{White}{255, 255, 255}
\declarecolor{Black}{0, 0, 0}
\declarecolor{Maroon}{128, 0, 0}
\declarecolor{Coral}{255, 127, 80}
\declarecolor{Red}{182, 21, 21}
\declarecolor{LimeGreen}{50, 205, 50}
\declarecolor{DarkGreen}{0, 80, 0}
\declarecolor{Purple}{146, 42, 158}
\declarecolor{Navy}{0, 0, 128}
\declarecolor{LightBlue}{84, 101, 202}
\usepackage{hyperref}
\hypersetup{
    colorlinks,
    citecolor=DarkGreen,
    linkcolor=LightBlue}
\usepackage[noabbrev, capitalise, nameinlink]{cleveref}

\newtheorem{theorem}{Theorem}
\newtheorem*{theorem*}{Theorem}
\newtheorem{lemma}{Lemma}

\newtheorem{corollary}{Corollary}
\newtheorem{example}{Example}
\newtheorem{assumption}{Assumption}

\newtheorem{proposition}{Proposition}

\newcommand{\reg}{\mathrm{Reg}}
\newcommand{\Dreg}{\mathrm{DReg}}

\usepackage{pifont}

\DeclareMathOperator*{\argmax}{argmax}
\DeclareMathOperator{\poly}{poly}

\def\+#1{\mathcal{#1}}
\def\-#1{\mathbb{#1}}

\newcommand{\notshow}[1]{{}}
\newcommand{\AutoAdjust}[3]{{\mathchoice{ \left #1 #2  \right #3}{#1 #2 #3}{#1 #2 #3}{#1 #2 #3}}}
\newcommand{\Xcomment}[1]{{}}

\newcommand{\InParentheses}[1]{\AutoAdjust{(}{#1}{)}}
\newcommand{\InBrackets}[1]{\AutoAdjust{[}{#1}{]}}
\newcommand{\InAngles}[1]{\AutoAdjust{\langle}{#1}{\rangle}}
\newcommand{\InNorms}[1]{\AutoAdjust{\|}{#1}{\|}}
\renewcommand{\part}[2]{\frac{\partial #1}{\partial #2}}

\newcommand{\hx}{\hat{x}}

\newcommand{\bx}{\overline{x}}

\newcommand{\tgap}{\textsc{TGap}}

\newcommand{\by}{\overline{y}}

\newcommand{\tu}{\Tilde{u}}
\newcommand{\tU}{\widetilde{U}}

\newcommand{\cyan}[1]{{\color{cyan}#1}}
\newcommand{\magenta}[1]{{\color{magenta}#1}}
\newcommand{\orange}[1]{{\color{orange}#1}}
\newcommand{\green}[1]{{\color[rgb]{0,0.7,0.2}#1}}
\newcommand{\bu}{\overline{u}}

\newcommand{\hu}{\hat{u}}
\newcommand{\hU}{\widehat{U}}

\newcommand{\AL}{\texttt{A2L}\xspace}
\newcommand{\OMWU}{\texttt{OMWU}\xspace}
\newcommand{\Bandit}{\texttt{Bandit}\xspace}

\title{From Average-Iterate to Last-Iterate Convergence in Games: A Reduction and Its Applications\thanks{Authors are ordered alphabetically.}}

\author{%
    \textnormal{Yang Cai} \\
    Yale University \\
    \texttt{yang.cai@yale.edu}\\
    \and 
    Haipeng Luo\\
    University of Southern California\\
    \texttt{haipengl@usc.edu}
    \and
    Chen-Yu Wei \\
    University of Virginia\\
    \texttt{chenyu.wei@virginia.edu}
    \and 
    Weiqiang Zheng \\
    Yale University\\
    \texttt{weiqiang.zheng@yale.edu}
}

\begin{document}
\maketitle
\begin{abstract}
    The convergence of online learning algorithms in games under self-play is a fundamental question in game theory and machine learning. Among various notions of convergence, last-iterate convergence is particularly desirable, as it reflects the actual decisions made by the learners and captures the day-to-day behavior of the learning dynamics. While many algorithms are known to converge in the average-iterate, achieving last-iterate convergence typically requires considerably more effort in both the design and the analysis of the algorithm. Somewhat surprisingly, we show in this paper that for a large family of games, there exists a simple black-box reduction that transforms the average iterates of an uncoupled learning dynamics into the last iterates of a new uncoupled learning dynamics, thus also providing a reduction from last-iterate convergence to average-iterate convergence. Our reduction applies to games where each player’s utility is linear in both their own strategy and the joint strategy of all opponents. This family includes two-player bimatrix games and generalizations such as multi-player polymatrix games. By applying our reduction to the Optimistic Multiplicative Weights Update algorithm, we obtain new state-of-the-art last-iterate convergence rates for uncoupled learning dynamics in multi-player zero-sum polymatrix games: (1) an $O(\frac{\log d}{T})$ last-iterate convergence rate under gradient feedback, representing an exponential improvement in the dependence on the dimension $d$ (i.e., the maximum number of actions available to either player); and (2) an $\Tilde{O}(d^{\frac{1}{5}}T^{-\frac{1}{5}})$ last-iterate convergence rate under bandit feedback, improving upon the previous best rates of $\Tilde{O}(\sqrt{d}T^{-\frac{1}{8}})$ and $\Tilde{O}(\sqrt{d}T^{-\frac{1}{6}})$. 
\end{abstract}

\section{Introduction}

The convergence of online learning algorithms in games under self-play is a fundamental question in both game theory and machine learning. Self-play methods for computing Nash equilibria have enabled the development of superhuman AI agents in competitive games such as Go~\citep{silver2017mastering}, Poker~\citep{bowling2015heads, brown2018superhuman, brown2019superhuman}, Stratego~\citep{perolat2022mastering}, and Diplomacy~\citep{meta2022human}. More recently, self-play learning algorithms have also been applied to large language model (LLM) alignment with human feedback, which can be modeled as a two-player zero-sum game~\citep{munos_nash_2023, swamy2024minimaximalist, ye2024theoretical, wu2024self, liu2024comal, liu2025statistical}. From a game-theoretic perspective, understanding the convergence behavior of self-play dynamics provides predictive insights into strategic multi-agent interactions and informs the design of more effective mechanisms. 

Among various notions of convergence, \emph{last-iterate convergence} is particularly desirable, as it reflects the actual decisions made by the learners and captures the day-to-day behavior of the learning dynamics. Formally, we say an algorithm has a last-iterate convergence rate of $f(T)$ to Nash equilibria (where $f(T)$ is a decreasing function with $\lim_{T \to +\infty} f(T) = 0$) if the generated sequence of strategy profiles $\{x^t\}$ satisfies that, for any $T \ge 1$, the iterate $x^T$ is an $f(T)$-approximate Nash equilibrium. This notion of \emph{anytime} last-iterate convergence is stronger than those considered in some prior works, which do not provide the same guarantee for every $T$, but instead require knowing the time horizon in advance. See \Cref{sec:related works} for a detailed discussion.

Despite its importance, classical results on the convergence of self-play dynamics primarily concern \emph{average-iterate} convergence. It is well known that when both players in a two-player zero-sum game employ no-regret online learning algorithms, the time-average of their strategies converges to a Nash equilibrium~\citep{freund1999adaptive}. In contrast, a large class of online learning algorithms—including Mirror Descent (MD) and Follow-the-Regularized-Leader (FTRL)—fails to achieve last-iterate convergence in such settings. Worse yet, the sequence of actual decisions made by the learners can diverge, exhibiting cyclic or even chaotic behavior~\citep{mertikopoulos2018cycles, bailey2018multiplicative, daskalakis2018limit}.

Aside from last-iterate convergence, another desirable property of self-play learning dynamics with multiple self-interested agents is \emph{uncoupledness}~\citep{daskalakis2011near}. In uncoupled learning dynamics, each player refines their strategies with minimal information about the game, while avoiding extensive inter-player coordination (e.g., shared randomness) and communication. Uncoupled learning dynamics are also studied under the names of independent or decentralized learning dynamics~\citep{ozdaglar2021independent}. They are particularly appealing in large-scale games, where either the game is high-dimensional and full information is hard to obtain, or there are many players and coordination is difficult to enforce.

Due to their importance in large-scale games, designing and analyzing uncoupled learning dynamics with last-iterate convergence rates has received extensive attention recently due to its importance both in theory and practice. 
We will review this long line of work in \Cref{sec:related works}, and only point out here that achieving last-iterate convergence typically requires considerably more effort in both the design and analysis of algorithms. On the design side, techniques such as \emph{optimism} and \emph{regularization} are often applied to ensure convergence of algorithms like OMD and FTRL. On the analysis side, establishing last-iterate convergence rates often demands problem-specific Lyapunov functions combined with tailored arguments. This stands in stark contrast to average-iterate convergence in two-player zero-sum games, which follows directly from the no-regret property of the algorithms.

\noindent\textbf{Our contribution}\;  Somewhat surprisingly, we show in this paper that for a large family of games, there exists a simple black-box reduction, $\AL$ (\Cref{alg:reduction}), that transforms the average iterates of an uncoupled learning dynamic into the last iterates of another uncoupled learning dynamic, thereby providing a reduction from last-iterate convergence to average-iterate convergence. Our reduction applies to games where each player’s utility is linear in both their own strategy and the joint strategy of all opponents. This family includes two-player bimatrix games and generalizations such as multi-player polymatrix games. Our reduction also works for games with nonlinear utilities but special structures such as the proportional response dynamics in Fisher markets (\Cref{app:PRD fisher}).

Aside from its conceptual contribution, our reduction also yields concrete improvements for uncoupled learning in multi-player zero-sum polymatrix games (which include two-player zero-sum games as special cases) and leads to new state-of-the-art last-iterate convergence rates. Let $d$ be the maximum number of actions available to each player. We apply our reduction to the Optimistic Multiplicative Weights Update (\OMWU) algorithm~\citep{rakhlin2013optimization, syrgkanis2015fast}, which is known to achieve an $O(\log d / T)$ average-iterate convergence rate in multi-player zero-sum polymatrix games under gradient feedback. As a result, our reduction yields uncoupled learning dynamics, \AL-\OMWU, that enjoys an $O(\log d / T)$ last-iterate convergence rate under gradient feedback. This represents an exponential improvement in the dependence on $d$ compared to the best previously known $O(\poly(d) / T)$ rate~\citep{cai2023doubly}. As an additional consequence, each player using \AL-\OMWU incurs only $O(\log d \log T)$ dynamic regret.

We further extend our reduction to the more challenging setting where only \emph{bandit feedback}, rather than full gradient feedback, is available. To this end, we design a new algorithm, \AL-\OMWU-\Bandit (\Cref{alg:bandit}), which augments the reduction with a utility estimation procedure. We show that \AL-\OMWU-\Bandit achieves a $\widetilde{O}(T^{-1/5})$ last-iterate convergence rate with high probability. This result improves upon the previously best known rates of $\widetilde{O}(T^{-1/8})$ (high probability) and $\widetilde{O}(T^{-1/6})$ (in expectation) established in~\citep{cai2023uncoupled} for the special case of two-player zero-sum games.

\subsection{Related Works}\label{sec:related works}
There is a vast literature on the regret guarantee and last-iterate convergence of learning algorithms in games. Here, we mainly review results applicable to two-player zero-sum games. For convergence under other conditions, such as strict equilibria and strong monotonicity, we refer readers to~\citep{giannou2021rate, jordan2024adaptive, ba2025doubly} and the references therein.

\noindent\textbf{Individual Regret Guarantee for Learning in Games}\; While $\Theta(\sqrt{T})$ regret is fundamental for online learning against adversarial loss sequences, improved regret is possible for multi-agent learning in games. Starting from the pioneer work of \citep{daskalakis2011near}, a long line of works propose uncoupled learning dynamics with $O(1)$ regret in zero-sum games~\citep{rakhlin2013optimization} and more generally variationally stable games~\citep{hsieh2021adaptive}, and $O(\poly(\log T))$ (swap) regret even in general-sum normal-form games~\citep{daskalakis2021near-optimal,anagnostides2022near-optimal,anagnostides2022uncoupled,soleymani2025faster} and Markov games~\citep{cai2024near,mao2024widetilde}. However, guarantees for the stronger (worst-case) \emph{dynamic regret} remain underexplored. It is worth noting that in the adversarial setting, achieving even sublinear dynamic regret is impossible. As a corollary of our $O(\log d \cdot T^{-1})$ last-iterate convergence rate, \AL-\OMWU guarantees that each player suffers only $O(\log d \cdot \log T)$ dynamic regret, an exponential improvement on the dependence on $d$~\citep{cai2023doubly}.

\noindent\textbf{Last-Iterate Convergence with Gradient Feedback}\;
Two classic methods that exhibit last-iterate convergence are the Extra Gradient (EG) algorithm~\citep{korpelevich_extragradient_1976} and the Optimistic Gradient (OG) algorithm~\citep{popov_modification_1980}, which are optimistic variants of vanilla gradient descent. Both algorithms are known to converge asymptotically in the last iterate, but their convergence rates remained open until recently, highlighting the challenge of analyzing last-iterate convergence rates. \cite{wei2021linear} established problem-dependent linear convergence rates for OG. Tight $O(1/\sqrt{T})$ last-iterate convergence rates for both EG and OG were subsequently established by~\citep{cai2022finite, gorbunov2022last}, matching the lower bounds from~\citep{golowich2020last, golowich2020tight}.

Although the $O(1/\sqrt{T})$ rate is tight for EG and OG, these algorithms are not optimal among first-order methods. By leveraging anchoring-based acceleration techniques from the optimization literature~\citep{halpern1967fixed}, accelerated versions of EG and OG have been proposed to achieve an $O(1/T)$ last-iterate convergence rate~\citep{diakonikolas2020halpern, yoon2021accelerated, cai2023doubly, cai_accelerated_2023, cai2024accelerated}. This matches the lower bounds for all first-order methods~\citep{ouyang2021lower, yoon2021accelerated}. We note that all these results incur a $\poly(d)$ dependence. In contrast, our algorithm achieves exponentially better dependence on $d$.

In addition to optimism, regularization is another effective technique for achieving last-iterate convergence. By adding a regularization term, the original game becomes strongly monotone, enabling standard gradient-based algorithms to achieve linear last-iterate convergence—albeit on the modified game. Setting the regularization strength to $O(\varepsilon)$ yields an iteration complexity of $O(\log(1/\varepsilon)/\varepsilon)$ for computing an $\varepsilon$-approximate Nash equilibrium~\citep{cen2021fast, cen2024fast}, corresponding to an $O(\log T / T)$ convergence rate. However, this approach has key limitations: it requires calibrating the regularization strength based on the total number of iterations $T$, and the resulting convergence guarantee applies only to the final iteration. It does not satisfy the stronger criterion of \emph{anytime} last-iterate convergence. This distinction is crucial: without the anytime guarantee, a trivial workaround exists—one can run any no-regret algorithm for $T-1$ steps and output the average iterate at step $T$. Moreover, only anytime guarantees give individual dynamic regret guarantees. 

While a diminishing regularization schedule does yield anytime last-iterate convergence, it leads to a slower convergence rate $\widetilde{O}(T^{-1/4})$ rate~\citep{park2023multi}. With an additional assumption on the regularized Nash equilibrium, \citep{zeng2022regularized} obtained a $O(T^{-1/3})$ last-iterate convergence rate. 
One might wonder whether more aggressive schedules, such as the doubling trick, could improve this. However, it remains unclear whether such methods guarantee true anytime convergence.\footnote{\citep{liupower} uses the doubling trick to establish a $\widetilde{O}(T^{-1})$ last-iterate convergence rate for the \emph{output policy} of the algorithm (which generally differs from the \emph{day-to-day policy} the players use to interact with each other). This is a weaker notion of last-iterate convergence than the one considered in this paper, which concerns the convergence of the day-to-day policy. } In contrast, our algorithm achieves anytime last-iterate convergence \emph{without} requiring prior knowledge of the time horizon, and it attains the optimal $O(1/T)$ rate, eliminating the extra logarithmic factor.

We also remark that there is a line of work on last-iterate convergence under \emph{noisy gradient} feedback~\citep{hsieh2022no}. \cite{abe2022mutation,abe2023last, wu2025learning} design algorithms based on perturbation and show last-iterate convergence to a stationary point near a Nash equilibrium in both gradient and noisy gradient feedback. \citep{abe2023last} also show asymptotic convergence to an exact Nash equilibrium by adaptively adjusting the perturbation. For non-asymptotic convergence rates, \citep{abe2024adaptively} establish a $\Tilde{O}(T^{-1/10})$ last-iterate convergence rate, which is improved to $\Tilde{O}(T^{-1/7})$ by~\citep{abe2025boosting} recently.

\noindent\textbf{Last-Iterate Convergence with Bandit Feedback}\;
Compared to gradient feedback, achieving last-iterate convergence under \emph{bandit feedback} (i.e., payoff-based feedback) is significantly more challenging and remains less understood. \cite{cai2023uncoupled} provides the first set of last-iterate convergence results for uncoupled learning dynamics in two-player zero-sum (Markov) games. They propose mirror-descent-based algorithms that achieve a high-probability convergence rate of $\widetilde{O}(T^{-1/8})$ and a rate of $\widetilde{O}(T^{-1/6})$ in expectation. Concurrently, \citep{chen2023finite, chen2024last} introduce best-response-based algorithms with an expected $\widetilde{O}(T^{-1/8})$ last-iterate convergence rate. A follow-up work~\citep{dong2024uncoupled} presents extensions to general monotone games. Our work improves upon these results by achieving a high-probability $\widetilde{O}(T^{-1/5})$ rate, and our guarantees extend to the more general setting of multi-player zero-sum polymatrix games.

\noindent\textbf{Last-Iterate Convergence of \OMWU} \;
The Optimistic Multiplicative Weights Update (\OMWU) algorithm~\citep{rakhlin2013optimization, syrgkanis2015fast} is a fundamental algorithm for learning in games. It achieves an $O(1/T)$ average-iterate convergence rate in two-player zero-sum games and guarantees $O(\log T)$ individual regret even in general-sum games~\citep{daskalakis2021near-optimal, anagnostides2022near-optimal, soleymani2025faster}. For last-iterate convergence in two-player zero-sum games, \cite{daskalakis2019last, hsieh2021adaptive} establish asymptotic convergence of \OMWU, and \cite{wei2021linear} proves a problem-dependent linear convergence rate under the assumption of a unique Nash equilibrium. However, for worst-case uniform last-iterate convergence—without relying on problem-dependent constants—\cite{cai2024fast} recently proved an $\Omega(1)$ lower bound, showing that \OMWU's last-iterate convergence can be arbitrarily slow.

\section{Preliminaries}

\textbf{Notations} 
We denote the $d$-dimensional probability simplex as $\Delta^{d} := \{x \in \-R^d: 0 \le x_i \le 1,  \sum_{i=1}^d x_i = 1\}$. The uniform distribution in $\Delta^d$ is denoted as $\mathrm{Uniform}(d)$.

\subsection{Games}

An $n$-player game $\+G = ([n], \{\+X_i\}, \{u_i\})$ consists of $n$ players indexed by $[n] := \{1, 2, \ldots, n\}$. Each player $i$ selects a strategy $x_i$ from a closed convex set $\+X_i \subseteq \mathbb{R}^{d_i}$. We refer to $d := \max_{i\in[n]} d_i$ as the dimensionality of the game. Given a strategy profile $x = (x_i, x_{-i})$, player $i$ receives utility $u_i(x) \in [0,1]$. A game is in normal-form if each player $i$ has a finite number of $d_i$ actions and given an action profile $a = (a_i, a_{-i})$, the utility for player $i$ is $u_i(a)$. In normal-form games, the strategy set for each player $u$ is the probability simplex $\+X_i = \Delta^{d_i}$, and given a strategy profile $x$, the expected utility for player $i$ is $\-E_{\forall i \in [n], a_i \sim x_i} \InBrackets{u_i(a_1, \ldots, a_n)}$.

We focus on games with \emph{linear utilities}.

\begin{assumption}\label{assumption:linear}
    A game $\+G$ has \emph{linear utilities} if for each player $i$, 
    \begin{itemize}
        \item the utility function $u_i(x_i, x_{-i})$ is linear in $x_i \in \+X_i$
        \item the utility function $u_i(x_i, x_{-i})$ is linear in $x_{-i} \in \times_{j\ne i} \+X_j$.
    \end{itemize}
\end{assumption}

For such games, we define $u_i(\cdot, x_{-i}) \in \mathbb{R}^{d_i}$ to be the unique vector such that $u_i(x) = \langle x_i, u_i(\cdot, x_{-i}) \rangle$. We remark that \Cref{assumption:linear} is weaker than assuming $u_i(x)$ is linear in the whole strategy profile $x \in \times_{i}\+X_i$. As an example, the bilinear function $u(x_1, x_2) = x_1^\top A x_2$ satisfies \Cref{assumption:linear} since it is linear in $x_1$ and also linear in $x_2$,  but it is not linear in $(x_1, x_2)$. This structure includes several important game classes:

\begin{example}[Two-Player Bimatrix Games]\label{ex:2p0m game}
   There are two players: the $x$-player chooses a mixed strategy $x \in \Delta^{d_1}$, and the $y$-player chooses $y \in \Delta^{d_2}$. Let $A, B \in \mathbb{R}^{d_1 \times d_2}$ be the payoff matrices. The $x$-player's utility is $u_x(x, y) = x^\top A y$, and the $y$-player's utility is $u_y(x, y) = x^\top B y$. Two-player bimatrix games satisfy \Cref{assumption:linear} due to the bilinear structure of the utilities. A \emph{two-player zero-sum game} is a special case where $A + B = 0$.
\end{example}

\begin{example}[Multi-Player Polymatrix Games]
    Polymatrix games~\citep{janovskaja1968equilibrium, howson1972equilibria} generalize bimatrix games to multiple players by introducing networked interactions. The $n$ players form the vertices of a graph $G = ([n], E)$, where each edge $(i, j) \in E$ represents a two-player bimatrix game between players $i$ and $j$ with payoff matrices $A_{i,j}$ and $A_{j,i}$. Each player $i$ selects a strategy from $\Delta^{d_i}$ and plays the same strategy against all neighbors. Given a strategy profile $x \in \times_i \Delta^{d_i}$, player $i$'s utility is defined as $\sum_{j : (i, j) \in E} x_i^\top A_{i,j} x_j$. A \emph{zero-sum polymatrix game}~\citep{cai2011minmax, cai2016zero} is a special case where the total utility sums to zero, i.e., $\sum_{i=1}^n u_i(x) = 0$ for any strategy profile $x$.
\end{example}

\noindent\textbf{Nash Equilibrium and Total Gap}\; An \emph{$\varepsilon$-approximate Nash equilibrium} of $\+G$ is a strategy profile $x \in \times_i \+X_i$ such that no player $i$ can deviate from $x_i$ and improve her utility by more than $\varepsilon > 0$,
\begin{align*}
    u_i(x_i, x_{-i}) \ge u_i(x'_i, x_{-i}) - \varepsilon, \forall i \in [n], \forall x'_i \in \+X_i.
\end{align*}
When $\varepsilon = 0$, the strategy profile is a Nash equilibrium.  We note that games with linear utilities have at least one Nash equilibrium~\citep{nash1950equilibrium}. Given a strategy profile $x$, we use the total gap to evaluate its proximity to Nash equilibria. Specifically, we define
$
    \tgap(x):= \sum_{i=1}^n \InParentheses{\max_{x_i' \in \+X_i} u_i(x_i', x_{-i}) - u_i(x)}.
$
We note that by definition, if $\tgap(x) \le \varepsilon$, then $x$ is an $\varepsilon$-approximate Nash equilibrium. In two-player zero-sum games, the total gap is also known as the duality gap. 

\subsection{Online Learning}
\noindent\textbf{Online Learning and Regret}\;
In an online learning problem, a learner repeatedly interacts with the environment. At each time $t \ge 1$, the learner chooses an action $x^t$ from a closed convex set $\mathcal{X}$, while the environment simultaneously selects a linear utility function. The learner then receives a reward $u^t(x^t) := \langle u^t, x^t \rangle$ and observes some feedback. We focus on the \emph{gradient feedback} setting, where the learner observes $u^t$; the more restricted \emph{bandit feedback} setting is studied in \Cref{sec:bandit}.

The goal of the learner is to minimize the \emph{(external) regret}, the difference between the cumulative utility and the utility of the best fixed action in hindsight, defined as
$
    \reg(T) := \max_{x \in \+X }\sum_{t=1}^T u^t(x) - \sum_{t=1}^T u^t(x^t).
$
An online learning algorithm is \emph{no-regret} if its regret is sublinear in $T$, that is, $\reg(T) = o(T)$. We note that classic results show that $\reg(T)=O(\sqrt{T})$ can be achieved and cannot be improved when utilities are chosen adversarially. 

A stronger notion of regret is the (worst-case) \emph{dynamic regret}~\citep{zinkevich2003online}, which competes with the utility achieved by the best actions in each iteration, defined as
$
    \Dreg(T) := \sum_{t=1}^T \max_{x \in \+X} u^t(x) - \sum_{t=1}^T u^t(x^t).
$
We remark that when utilities are chosen adversarially, it is impossible to achieve sublinear dynamic regret, that is, $\Dreg(T) = \Omega(T)$.

\noindent\textbf{Convergence of Learning Dynamics} \;
For multi-agent learning dynamics in games, each player uses an online learning algorithm to repeatedly interact with other players. At each time $t$, player $i$ chooses strategy $x^t_i$, gets utility $u_i(x^t_i, x^t_{-i})$, and receives the utility vector $u^t_i := u_i(\cdot, x^t_{-i})$ as feedback. For any $T \ge 1$,  we denote the average iterate as $\bx^T:= (\bx^T_1, \ldots, \bx^T_n)$ where $\bx^T_i = \frac{1}{T}\sum_{t=1}^T x^t_i$, and each player $i$'s individual regret as $\reg_i(T)$. We show that the average iterate $\bx^T = \frac{1}{T}\sum_{t=1}^T x^t$ has a total gap bounded by $\frac{1}{T}\sum_{i=1}^n \reg_i(T)$ for zero-sum polymatrix games (a generalization of the classic result of~\citep{freund1999adaptive} for zero-sum bimatrix games).  Thus, as long as each player has sublinear regret, average-iterate convergence to Nash equilibria is guaranteed. The proof is in \Cref{sec:proof average-iterate}.

\begin{lemma}[Average-Iterate Convergence by Bounding Regret]\label{lemma:average-iterate}
    Let $\{x^t\}$ be the iterates of an online learning dynamics in a zero-sum polymatrix game. Define $\bx^T = \frac{1}{T}\sum_{t=1}^T x^t$ to be the average iterate for all $T \ge 1$. Then the total gap of the average iterate $\bx^T$ for any $T \ge 1$ is 
    \begin{align*}
        \tgap(\bx^T) := \max_{x\in \times \Delta^{d_i}}\sum_{i=1}^n u_i(x_i, \bx^T_{-i}) - u_i(\bx^T) = \frac{1}{T} \sum_{i=1}^n \reg_i(T).
    \end{align*}
\end{lemma}
However, the convergence of the average-iterate sequence $\{\bx^t\}$ does not imply the convergence of the last-iterate $\{x^t\}$, which is the real strategy played by the agents. Even worse, many online learning algorithms \emph{diverge} and exhibit \emph{cycling} or \emph{chaotic behavior} even in simple two-player zero-sum games~\citep{mertikopoulos2018cycles, bailey2018multiplicative, daskalakis2018limit}.

\noindent\textbf{Optimistic Multiplicative Weights Update (OMWU)}\; The OMWU algorithm~\citep{rakhlin2013optimization, syrgkanis2015fast} is an optimistic variant of the classic Multiplicative Weights Update (MWU) algorithm~\citep{arora2012multiplicative} with the only modification that the most recent utility observation is counted twice. In each time $t \ge 1$, OMWU with steps size $\eta > 0$ chooses the strategy $x^t \in \Delta^d$ to be
\begin{align*}
    x^t = \argmax_{x \in \Delta^d} \left \{\InAngles{x, \sum_{k < t} u^k + u^{t-1}} - \frac{1}{\eta} \phi(x) \right \},
\end{align*}
where we define $u^0: = 0$ and $\phi$ is the negative entropy regularizer. The OMWU updates also admits a closed-form expression:
\begin{align*}
    x^t[i] = \frac{\exp\InParentheses{\eta\InParentheses{\sum_{k<t} u^k[i] + u^{t-1}[i]}}}{\sum_{j=1}^{d} \exp\InParentheses{\eta \InParentheses{\sum_{k<t} u^k[j] + u^{t-1}[j]}} }, \forall i \in [d].
\end{align*}

OMWU has been extensively studied in the online learning and game theory literature~\citep{rakhlin2013optimization, syrgkanis2015fast, daskalakis2018training, chen2020hedging, wei2021linear, daskalakis2021near-optimal, cai2024fast, soleymani2025faster}, with state-of-the-art individual regret guarantees in zero-sum games and general-sum games. We will use the \emph{regret bounded by variation in utility} (RVU) property of OMWU throughout the paper. 
\begin{theorem}[Proposition 7 of \citep{syrgkanis2015fast}]\label{thm:RVU}
    The regret of OMWU for a sequence of utilities $\{u^t\}$ satisfies 
    \begin{align*}
       \reg(T) = \max_{x \in \Delta^d }\sum_{t=1}^T \InAngles{u^t, x - x^t} \le \frac{\log d}{\eta} + \eta \sum_{t=1}^T \InNorms{u^t - u^{t-1}}_\infty^2 - \frac{1}{4\eta} \sum_{t=1}^T \InNorms{x^t - x^{t-1}}_1^2.
    \end{align*}
\end{theorem}
\begin{lemma}[Adapted from Theorem 4 of \citep{syrgkanis2015fast}]\label{lemma:utility to action}
    In a $n$-player normal-form game, let $x$ and $x'$ be two strategy profiles. Then $\sum_{i=1}^n \InNorms{u_i(\cdot, x_{-i}) - u_i(\cdot, x'_{-i})}_\infty^2 \le (n-1)^2 \sum_{i=1}^n \InNorms{x_i - x'_i}_1^2$.
\end{lemma}

\section{A Reduction From Last-Iterate to Average-Iterate}\label{sec:reduction}
In this section, we introduce a black-box reduction that transforms any online learning algorithm into a new one such that, when employed in a game with linear utilities (\Cref{assumption:linear}), the produced iterates exactly match the averaged iterates produced by the original algorithm. 

Specifically, consider an $n$-player game $\+G$, where each player $i$ employs an online learning algorithm $\+R_i$ and generate iterates $\{\hx^t\}$. We propose a reduction, \AL (\Cref{alg:reduction}), and let each player $i$ employs $\AL(\+R_i)$, generating a new sequence $\{\bx^t\}$. We show that for any $t \ge 1$, $\bx^t = \frac{1}{k}\sum_{k=1}^t \hx^k$; that is, the last iterate of  $\{\bx^t\}$ equals the running average of $\{\hx^t\}$.

The intuition of the reduction is as follows. Given any algorithm $\+R$, $\AL(\+R)$ runs $\+R$ as a subroutine and gets a strategy $x^t$ in every iteration $t$. We note that the sequence $\{x^t\}$ is produced internally in the reduction and not played by the players. $\AL(\+R)$ plays the averaged strategy $\bx^t = \frac{1}{t}\sum_{k=1}^t x^k$ and gets the gradient feedback $\bu^t$. To show that the last iterate of  $\{\bx^t\}$ equals the running average of $\{\hx^t\}$, it suffices to prove that $\hx^t = x^t$.  The next step is the key observation: when all players employ this reduction and play the running average, they can recover/extract the unseen utility vector $u^t = u_i(\cdot, x^t)$ under strategy $x^t$ (which is not played) by the identity $u^t = t \bu^t - (t-1) \bu^{t-1}$, which holds due to the linearity of the utility (\Cref{assumption:linear}). Then $\AL(\+R)$ forwards this recovered utility to $\+R$. Since $\+R$ receives the same utilities when played alone or as a subroutine in $\AL(\+R)$, we can conclude that the iterates $\{\hx^t\}$ equal the internal iterates of $\{x^t\}$. Thus $\bx^t = \frac{1}{t}\sum_{k=1}^t x^k = \frac{1}{t}\sum_{k=1}^t \hx^k$. A formal proof using induction is provided later in this section. {We also provide an illustration of the \AL reduction for a two-player zero-sum game $\max_x\min_y x Ay$ (\Cref{ex:2p0m game}) below, where we use the subscripts $x$ and $y$ to denote the corresponding quantities for the two players.}

\begin{center}
\begin{tikzpicture}[scale=0.9]

  \draw[->] (-2.8, .4) -- (-2, .4) node[above left] {$u_x^{t-1}$};
  \draw[->] (-2.8, -.6) -- (-2, -.6) node[above left] {$u^{t-1}_y$};

  \draw[thick] (-2, 0) rectangle (-0.8, .8);
  \node at (-1.4, .4) {$\+R_x$};
  \draw[->] (-0.8, .4) node[above right] {\magenta{$x^t$}} -- (0.0, .4);

  \draw[thick] (0.0, 0) rectangle (1.8, .8);
  \node at (0.9, .4) {\textbf{Average}};
  \draw[->] (1.8, .4) -- (2.6, .4) node[midway, above] {\green{$\bar{x}^t$}};

  \draw[thick] (-2, -1) rectangle (-0.8, -0.2);
  \node at (-1.4, -.6) {$\+R_y$};
  \draw[->] (-0.8, -.6) node[above right] {\cyan{$y^t$}} -- (0.0, -.6);

  \draw[thick] (0.0, -1) rectangle (1.8, -0.2);
  \node at (0.9, -.6) {\textbf{Average}};
  \draw[->] (1.8, -.6) -- (2.6, -.6) node[midway, above] {\orange{$\bar{y}^t$}};

  \draw[thick] (2.6, -.8) rectangle (3.0, .6);
  \node at (2.8, -.1) {$A$};

  \draw[->] (3.0, .4) -- (3.2, .4) -- (3.8, -.6) -- (5.4, -.6)
      node[above left] {\scriptsize $\green{\bu_y^{t}} = -A^\top \green{\bar{x}^t}$};

  \draw[->] (3.0, -.6) -- (3.2, -.6) -- (3.8, .4) -- (5.4, .4)
      node[above left] {$\orange{\bu_x^{t}}=A \orange{\bar{y}^t}$};

  \draw[thick] (5.4, 0) rectangle (7.0, .8);
  \node at (6.2, .4) {\textbf{Extract}};

  \draw[thick] (5.4, -1) rectangle (7.0, -0.2);
  \node at (6.2, -.6) {\textbf{Extract}};

  \draw[->] (7.1, .4) node[above right] {\cyan{$u^t_x$}} -- (7.8, .4);
  \draw[thick] (7.8, 0) rectangle (9.0, .8);
  \node at (8.4, .4) {$\+R_x$};
  \draw[->] (9.0, .4) node[above right] {\magenta{$x^{t+1}$}} -- (9.8, .4);

  \draw[->] (7.1, -.6) node[above right] {\magenta{$u^{t}_y$}} -- (7.8, -.6);
  \draw[thick] (7.8, -1) rectangle (9.0, -0.2);
  \node at (8.4, -.6) {$\+R_y$};
  \draw[->] (9.0, -.6) node[above right] {\cyan{$y^{t+1}$}} -- (9.8, -.6);

  \node at (10.6, -0.1) {$\cdots$};
  \node at (-3.7, -0.1) {$\cdots$};

  \draw[dashed, rounded corners] (-0.2,-1.5) rectangle (7.2,1.3);
  \node[above=1pt] at (3.5,-1.5) {\textbf{\AL: Black Box}};
\end{tikzpicture}
\end{center}

A key feature of this reduction is that it preserves uncoupledness: each player only observes their own utility, without access to other players’ strategies, or any form of communication or shared randomness.

\begin{algorithm}[!ht]
    \KwIn{An online learning algorithm $\+R$} 
    \caption{\AL: Black-Box Reduction from Average to Last-Iterate}\label{alg:reduction}
    \For{$t = 1,2, \ldots,$}{
    Get $x^t$ from $\+R$ \\
    Play $\bar{x}^t = \frac{1}{t}\sum_{k=1}^t x^k$ and receive $\bu^t$ \label{line3} \Comment{Play the running average iterates}
    Compute $u^t = t \cdot \bu^t - \sum_{k=1}^{t-1} u^k$\label{line4} \Comment{Recover the unseen utility}
    Forward utility $u^t$ to $\+R$ 
    }
\end{algorithm}

\begin{theorem}\label{thm:reduction}
    Consider any $n$-player game $\+G$ that satisfies \Cref{assumption:linear}). Let \( \{\hat{x}^t_i : i \in [n], t \in [T]\} \) be the iterates generated by each player $i$ running \( \+R_i \), and let \( \{\bar{x}^t_i : i \in [n], t \in [T]\} \) be the iterates generated by running \AL\( (\+R_i) \). Then for all $i \in [n]$ and $t \in [T]$, we have $
        \bar{x}^t_i = \frac{1}{t} \sum_{k=1}^t \hat{x}^k_i$.
\end{theorem}

\begin{proof}[Proof of \Cref{thm:reduction}]
    Note that \AL($\+R_i$) uses $\+R_i$ as a subroutine. Let \( \{x^t_i\} \) denote the strategies produced by $\+R_i$ during the execution of \AL($\+R_i$). By \Cref{line3}, we have $
        \bar{x}^t_i = \frac{1}{t} \sum_{k=1}^t x^k_i, \forall i \in [n],  t \in [T].$
    Hence, it suffices to show that $x^t_i = \hat{x}^t_i$ for all $i \in [n], t \in [T]$.

    We prove this by induction on $t$. Let the induction hypothesis be: for all $k \le t$ and all $i \in [n]$, we have $x^k_i = \hat{x}^k_i$ and $u^k_i = u_i(\cdot, x^k_{-i})$.

    \textbf{Base Case: $t = 1$.} Initialization ensures $x^1_i = \hat{x}^1_i$ for all $i \in [n]$. By \Cref{line3,line4}, we also have $u^1_i = \bu^1_i = u_i(\cdot, \bx^1_{-i}) = u_i(\cdot, x^1_{-i}) = u_i(\cdot, \hx^1_{-i})$.

    \textbf{Induction Step:} Suppose the hypothesis holds for $t$. Then for all $k \le t$ and $i \in [n]$, we have $u^k_i = u_i(\cdot, x^k_{-i}) = u_i(\cdot, \hat{x}^k_{-i})$. Since $\+R_i$'s output at time $t+1$ depends only on the sequence $\{u^k_i\}_{1 \le k \le t}$, we get $x^{t+1}_i = \hat{x}^{t+1}_i$ for all $i$.

    At time $t+1$, each player $i$ in \AL($\+R_i$) plays
    $ \bar{x}^{t+1}_i = \frac{1}{t+1} \sum_{k=1}^{t+1} x^k_i$.
    Then each player $i$ receives the utility vector
    \begin{align*}
        \bu^{t+1}_i = u_i\left(\cdot, \frac{1}{t+1} \sum_{k=1}^{t+1} x^k_{-i}\right) = \frac{1}{t+1} \sum_{k=1}^{t+1} u_i(\cdot, x^k_{-i}), \quad \text{(by linearity of $u_i$ (\Cref{assumption:linear}))}.
    \end{align*}
    Then,
    $
        u^{t+1}_i = (t+1) \cdot \bu^{t+1}_i - \sum_{k=1}^{t} u^k_i = \sum_{k=1}^{t+1} u_i(\cdot, x^k_{-i}) - \sum_{k=1}^{t} u_i(\cdot, x^k_{-i}) = u_i(\cdot, x^{t+1}_{-i}).
    $
    Thus, the induction hypothesis holds for $t+1$.

    By induction, $x^t_i = \hat{x}^t_i$ for all $i$ and $t$, and so
    $ \bar{x}^t_i = \frac{1}{t} \sum_{k=1}^t x^k_i = \frac{1}{t} \sum_{k=1}^t \hat{x}^k_i. $
\end{proof}

\noindent\textbf{Discussion on Limitations}\; We discuss some limitations of our reduction compared to existing algorithms with last-iterate convergence guarantees, such as Extra Gradient (EG), Optimistic Gradient (OG), and their variants. Our reduction relies on the linear utility assumption (\Cref{assumption:linear}), which includes important classes of games such as multi-player polymatrix games, but does not extend to settings with concave utilities—such as convex-concave min-max optimization—where EG and OG are known to achieve last-iterate convergence~\citep{cai2022finite}. Moreover, since our approach reduces last-iterate convergence to average-iterate convergence, it inherently requires computing the running average of iterates. In contrast, algorithms like EG and OG do not involve such averaging procedures. Nonetheless, our reduction offers a simple and broadly applicable method for achieving last-iterate convergence in games with linear utilities, and it yields state-of-the-art convergence rates under both gradient and bandit feedback settings (see \Cref{sec:gradient,sec:bandit}).

{
\noindent\textbf{Extension to Games with Nonlinear Utilities}\; The core idea of \AL reduction lies in extracting the unseen feedback of $u^t$ from the observed feedback $\bu^t$, for which \Cref{assumption:linear} is sufficient but not necessary. 
For example, in \Cref{app:PRD fisher},
we show that \AL is also applicable to Proportional Response Dynamics (PRD) in Fisher markets where \Cref{assumption:linear} does not hold. We obtain $O(1/T)$ last-iterate convergence rate for PRD in markets with Gross Substitutes utilities, improving~\citep{cheung2025proportional} which only shows an average-iterate convergence. We defer a detailed background on Fisher market and discussion on why \AL works to \Cref{app:PRD fisher}.
We believe that  there are other similar examples where \AL applies even without \Cref{assumption:linear}.

\noindent\textbf{Comparision with \citep{cutkosky2019anytime}}\;
Our work is related to \citep{cutkosky2019anytime}, which 
proposes a reduction called \emph{anytime online-to-batch (\texttt{AO2B})} and achieves last-iterate convergence in offline convex optimization. While our \AL reduction and the \texttt{AO2B} reduction share the idea of playing the average of previous iterates, \emph{they differ significantly in several key aspects}:
\begin{itemize}[leftmargin=*]
    \item \textbf{Problem setting:} \texttt{AO2B} is designed for static optimization, where the loss function remains fixed across rounds. \AL, on the other hand, applies to a dynamic multi-agent game setting, where each player’s loss function changes over time due to the changing actions of other players.
    \item \textbf{Objective:} \texttt{AO2B} aims to solve a single offline optimization problem with improved convergence rates and does not guarantee equivalence between the last iterate of the new learning algorithm and the average iterate of the original algorithm (since the gradient feedbacks are different). In contrast, \AL aims to establish the equivalence between the last iterate of the new dynamics and the averaged iterate of the original dynamics.
    \item \textbf{Use of gradient:} In \texttt{AO2B}, the learner updates directly using the gradient at the average iterate $\bx^t = \frac{1}{t}\sum_{k=1}^t x^k$. In contrast, \AL requires post-processing the gradient feedback to recover the true gradient at the original iterate $x^t$, as shown in Line 4 of \Cref{alg:reduction}.
\end{itemize}

\noindent\textbf{Weighted Version of \AL}\; The \AL reduction naturally extends to the weighted averaging setting, provided all players agree on the weights $\{\alpha_t\}$ in advance and incorporate them into the reduction. More specifically, in each iteration $t$, each player plays the weighted average of their past iterates $\bx^t_i = \frac{1}{\sum_{k=1}^t \alpha_k} \sum_{k=1}^t \alpha_k x^k_i$.  Due to the linearity of utilities, the gradient feedback can still be used to recover the gradient corresponding to the original iterate $x^k$, and the \AL reduction in \Cref{alg:reduction} is the special case of $\{\alpha_t = 1\}$. As a result, the last iterate of the new dynamics corresponds to the weighted average in the original dynamics. The correctness follows the same steps as in \Cref{thm:reduction}. This weighted version of A2L provides additional flexibility and can be beneficial in practice. For example, for regret matching~\citep{tammelin_solving_2014}, an algorithm widely used for solving large-scale games like poker~\citep{brown2018superhuman, brown2019superhuman}, linear averaging (i.e., $\alpha_t = t$) often empirically outperforms uniform averaging, despite having the same theoretical rate. 
}

\section{Learning in Zero-Sum Games with Gradient Feedback}
\label{sec:gradient}

In this section, we show that our reduction gives uncoupled learning dynamics with improved last-iterate convergence rates in two-player zero-sum games and zero-sum polymatrix games. 

We apply our reduction to the Optimistic Multiplicative Weights Update (OMWU) algorithm, which has been extensively studied in the literature~\citep{rakhlin2013optimization, syrgkanis2015fast, daskalakis2018training, chen2020hedging, wei2021linear, daskalakis2021near-optimal, cai2024fast, soleymani2025faster}. For $d$-dimensional two-player zero-sum games, while a recent result~\citep{cai2024fast} shows that OMWU's last-iterate convergence can be arbitrarily slow, it is well-known that OMWU achieves an $O(\frac{\log d}{T})$ average-iterate convergence rate~\citep{rakhlin2013optimization}. Now let us denote the algorithm after \AL reduction as \AL-\OMWU.   By \Cref{thm:reduction} and \Cref{lemma:average-iterate}, we get $O(\log d\cdot T^{-1})$ last-iterate convergence rate result of  \AL-\OMWU (\Cref{thm:last-gradient}). As a corollary of our fast last-iterate convergence rates, \AL-\OMWU also guarantees that each player's individual dynamic regret is only $O(\log d \log T)$.
Missing proofs in this section are in \Cref{sec:proof last-gradient}. 
\begin{theorem}\label{thm:last-gradient}
    Let $\{\bx^t\}$ be the iterates of \AL-\OMWU learning dynamics in an $n$-player zero-sum polymatrix game with step size $\eta \le \frac{1}{2(n-1)}$. Then for any $T \ge 1$, $\bx^T$ is an $(\frac{\sum_{i=1}^n \log d_i}{\eta T})$-approximate Nash equilibrium.
\end{theorem}

\begin{corollary}\label{corollary:dynamic}
    In the same setup as \Cref{thm:last-gradient}, for any player $i \in [n]$, $\Dreg_i^T = O\InParentheses{\frac{\sum_{i=1}^n \log d_i}{\eta} \log T}$.
\end{corollary}

To our knowledge, our results give the fastest last-iterate convergence rates for learning in zero-sum games. Compared to the accelerated optimistic gradient (AOG) algorithm~\citep{cai2023doubly} which has $O(\frac{\poly(d)}{T})$ last-iterate convergence rates, our results achieves the same optimal $\frac{1}{T}$ dependence while improve exponentially on the dependence on the dimension $d$. Compared to the entropy regularized extragradient algorithm~\citep{cen2021fast,cen2024fast}, which achieves $O(\frac{\log d \log T}{T})$ last-iterate convergence rate, our results do not require regularization and shave the $\log T$ factor. Moreover, our results and those of~\citep{cai2023doubly} are stronger \emph{anytime convergence results} that holds for every iterate $T \ge 1$, while results in \citep{cen2021fast,cen2024fast} hold only for the final iterations and require knowing the total number of iterations $T$ in advance for tuning the regularization parameter (as mentioned in Section~\ref{sec:related works}, a trivial solution exists for such weaker guarantees), thus with no dynamic regret guarantee.

\noindent\textbf{Robustness against Adversaries}\; In the case where other players are adversarial and may not follow the \AL-\OMWU algorithm, a simple modification can ensure sublinear regret for the player. Specifically, the player can track their cumulative utilities and regret up to time $T$, and if the regret exceeds $\Omega(\log T)$—which cannot happen if all players follow \AL-\OMWU—they can switch to a standard no-regret algorithm that guarantees $O(\sqrt{T})$ regret in the worst case.

\section{Learning in Zero-Sum Games with Bandit Feedback}\label{sec:bandit}

In this section, we show the application of our reduction for uncoupled learning dynamics in zero-sum games with bandit feedback. Unlike the gradient feedback setting, where each player observes the utility of every action, the bandit feedback setting is more realistic and much harder. Specifically, we consider the following uncoupled interaction protocol: in each iteration $t$
\begin{itemize}
    \item Each player $i$ maintains a mixed strategy $x^t_i \in \Delta^{d_i}$ and plays an action $a^t_i \sim x^t_i$;
    \item Each player $i$ then receives a number $u_i(a^t)$ as feedback.
\end{itemize}
Importantly, a player only observes the utility of one action but not the whole utility vector. Also, a player does not observe the actions/strategies of other players.

\noindent\textbf{Algorithm Design}\; Ideally, we would like to have a bandit learning dynamics whose last iterates are the average iterates of another bandit learning dynamics. This would give a $O(T^{-\frac{1}{2}})$ last-iterate convergence rate since bandit algorithms have $O(\sqrt{T})$ regret. However, the lack of utility vector feedback prevents us from directly applying \AL reduction since we can no longer use the linearity to recover the original utility when other players use the averaged strategies. To overcome this challenge, we design \AL-\OMWU-\Bandit (\Cref{alg:bandit}), where in addition to \AL, we add a new component for estimating the utility vector from bandit feedback. With a sufficiently accurate estimation of the utility vector, we can then use \AL as in the gradient feedback setting.

\begin{algorithm}[!ht]
    \caption{\AL-\OMWU-\Bandit: \OMWU with bandit feedback and \AL{} reduction (for player $i$)}\label{alg:bandit}
    \textbf{Parameters:} Step size $\eta > 0$,  $\{B_t \ge t^4\}$, $\{\varepsilon_t = t^{-1}\}$ \\
    \textbf{Initialization:} $x^1_i \leftarrow \mathrm{Uniform}(d_i))$; $\hU^0_i = 0$.\\
    \For{$t = 1,2, \ldots,$}{
    $\bx^t_i \leftarrow \frac{1}{t}\sum_{k=1}^t x^k_i$  \\
    $\bx^t_{i,\varepsilon} \leftarrow (1-\varepsilon_t) \bx_i^t + \varepsilon_t \mathrm{Uniform}(d_i)$ \label{line:mixing}\\
    Sample actions from $\bx^t_{i,\varepsilon}$ for $B_t$ rounds
    and compute $\widehat{U}_{i,\varepsilon}^t$ as an estimate of $\bu_{i,\varepsilon}^t$ based on \eqref{eq:estimator}  \label{line:sample}\\
    Computed estimated utility $\hu_{i,\varepsilon}^t \leftarrow t \cdot \hU_{i,\varepsilon}^t - (t-1)\cdot \hU_{i,\varepsilon}^{t-1}$ \\
    Update using \OMWU: $ x_i^{t+1} \leftarrow \argmax_{x \in \Delta^{d_i}} \left \{\InAngles{x, \sum_{k \le t} \hu_{i,\varepsilon}^k + \hu_{i,\varepsilon}^{t}} - \frac{1}{\eta} \phi(x) \right \}$.\label{line:OMWU update}
    }
\end{algorithm}

\noindent\textbf{Estimation}\; To estimate the utility vector when each player only observes the reward after taking an action, we let the each player use their current strategy to interact with each other for a sequence of $B_t$ rounds and then use the bandit feedback to calculate an estimated utility (\Cref{line:sample}). Let us denote the sampled actions over the $B_t$ rounds as $\{a^1,\ldots, a^{B_t}\}$ and the corresponding bandit feedback as $\{r^1, \ldots, r^{B_T}\}$. We use the following simple estimator for the utility vector:
\begin{align}\label{eq:estimator}
    \hU_{i,\varepsilon}^t[a] = \frac{\sum_{k =1}^{B_t} r^k \cdot \mathbb{I}[a^k = a]}{ \sum_{k =1}^{B_t} \mathbb{I}[a^k = a]}, \forall a.
\end{align}
To encourage exploration, we also mix the strategy with $\varepsilon_t$ amount of uniform distribution (\Cref{line:mixing}) to ensure that every action receives sufficient samples in each epoch. 

{\noindent\textbf{Last-Iterate Convergence Rate}\;  
In \Cref{alg:bandit}, each iterate \( \{\bx^t_\varepsilon\} \) is repeated \( B_t \) times within epoch \( t \). We will show that this sequence \( \{\bx^t_\varepsilon\} \) enjoys a \( \widetilde{O}(t^{-1}) \) last-iterate convergence rate. Let \( \{\by^k\} \) denote the actual iterates executed by the players in every round \( k \), which remains the same in each epoch. When $B_t = d\cdot t^4$, we get the following $\Tilde{O}(d^{\frac{1}{5}}k^{-\frac{1}{5}})$ last-iterate rate of \( \{\by^k\} \). If the maximal number of actions $d$ is unknown to the players in a fully uncoupled setting, we can choose $B_t = t^4$ and get an $\Tilde{O}(dk^{-\frac{1}{5}})$ convergence rate.
 
\begin{theorem}\label{thm:last-bandit}
   Consider an $n$-player zero-sum polymatrix game and denote $d := \max_{i \in [n]} d_i$.  
   Let $\{\by^k\}$ be the iterates generated by all players running \AL-\OMWU-\Bandit with $\eta \le \frac{1}{6n}$, then for any $\delta\in (0,1)$,  with probability at least $1 - O(nd^3\delta)$, for any $k \ge 1$, $\by^k$ is an $\beta_k$-approximate Nash equilibrium with
    \begin{align*}
        &\beta_k = O\InParentheses{\frac{n\log^2(\frac{dk}{\delta})}{\eta}\cdot d^{\frac{1}{5}}k^{-\frac{1}{5}}} \text{ when } B_t = d\cdot t^4 \;\text{ or }\; \beta_k = O\InParentheses{\frac{n\log^2(\frac{dk}{\delta})}{\eta}\cdot dk^{-\frac{1}{5}}}\text{ when } B_t = t^4.
    \end{align*}
\end{theorem}}
To our knowledge, $\Tilde{O}(T^{-\frac{1}{5}})$ is the fastest last-iterate convergence rate of uncoupled learning dynamics in zero-sum polymatrix games, improving the previous $\Tilde{O}(T^{-\frac{1}{8}})$ (high-probability) and $\Tilde{O}(T^{-\frac{1}{6}})$ (expectation) rates for two-player zero-sum games~\citep{cai2023uncoupled}.

\noindent\textbf{Technical highlight}\; The main challenge in analyzing the last-iterate convergence rate of \Cref{alg:bandit} lies in handling the estimation error. Let $\Delta^t_i = \hU^t_{i,\varepsilon} - \bu^t_{i,\varepsilon}$ denote the estimation error from \Cref{line:sample}. A naive analysis leads to a summation of \emph{first-order} error terms, $O(\sum_{t=1}^T \| t \Delta^t_i \|)$, which results in a suboptimal $\Tilde{O}(T^{-1/7})$ rate. To improve this, we introduce a refined analysis that carefully balances the error terms with the negative terms in the RVU inequality (\Cref{thm:RVU}), resulting in a summation of \emph{second-order} error terms, $O(\sum_{t=1}^T \| t \Delta^t_i \|^2)$ (see \Cref{lemma:regret with error}). This sharper analysis ultimately yields an improved $\Tilde{O}(T^{-1/5})$ convergence rate. The full proof is provided in \Cref{sec:proof bandit}.

\noindent\textbf{Discussion}\;
Our results also reveal an intriguing distinction between the gradient and bandit feedback settings. In the gradient feedback case, while our reduction improves the dependence on the dimension $d$, it achieves the same $O(T^{-1})$ convergence rate as existing algorithms such as the accelerated optimistic gradient (AOG) algorithm~\citep{cai2023doubly}. This might suggest that applying a similar utility estimation procedure to AOG would also yield an $\Tilde{O}(T^{-1/5})$ rate under bandit feedback. However, the analyses of these two types of algorithms diverge significantly in the presence of estimation error:

For AOG, establishing last-iterate convergence requires proving the (approximate) monotonicity of a carefully designed potential function at \emph{every iteration}. This sensitivity to estimation error leads to error propagation and ultimately a slower rate: we did not manage to get a rate faster than $\Tilde{O}(T^{-1/8})$~\citep{cai2023uncoupled}. In contrast, our reduction—by transforming last-iterate convergence into average-iterate convergence—only requires analysis on the regret, which is simpler, less sensitive, and gives the improved $\Tilde{O}(T^{-1/5})$ rate. This suggests an advantage of our \AL reduction over existing algorithms in the bandit feedback setting. Nevertheless, it remains an interesting open question whether one can design uncoupled learning dynamics with fast convergence rates in the bandit feedback setting using existing algorithms like AOG without relying on the \AL reduction.

\noindent\textbf{Robustness against Adversaries}\; When facing possibly adversarial opponents who may not follow \Cref{alg:bandit}, the player could still guarantee sublinear regret by running \Cref{alg:bandit} with an additional regret estimation procedure. Whenever the player detects that her regret is $\Omega(T^{\frac{4}{5}})$---which can not happen when all the players employ \Cref{alg:bandit}---she can switch to a standard bandit algorithm with optimal regret. This modification guarantees $\widetilde{O}(T^{\frac{4}{5}})$ regret in the worst-case. We present the regret estimation procedure and the proof in \Cref{app:bandit robustness}.
\vspace{-5pt}
\section{Conclusion}
In this paper, we give a simple black-box reduction, \AL, that transforms the average iterates of an uncoupled learning dynamics to the last iterates of a new uncoupled learning dynamics. By \AL reduction, we present improved last-iterate convergence rates of uncoupled learning dynamics in zero-sum polymatrix games: (1) $O(\log d \cdot T^{-1})$ rate under gradient feedback, and (2) $\Tilde{O}(d^{\frac{1}{5}} T^{-\frac{1}{5}})$ rate under bandit feedback. It is an interesting future direction to explore further applications of the \AL reduction for learning in games. Other directions include designing uncoupled learning dynamics for more general games like Markov games and time-varying games with faster last-iterate convergence rates under bandit feedback. 

\subsection*{Acknowledgement}
{We thank the anonymous NeurIPS reviewers for constructive comments. We also thank Arnab Maiti and Nathan Monette for comments on typos and minor errors in the previous version of the paper, which helped us improve the paper.}
YC is supported by the NSF Awards CCF-1942583 (CAREER) and CCF-2342642. HL is supported by NSF award IIS-1943607. WZ is supported by the NSF Awards CCF-1942583 (CAREER),  CCF-2342642, and a Research Fellowship from the Center for Algorithms, Data, and Market Design at Yale (CADMY).

\printbibliography

@inproceedings{birnbaum2011distributed,
  title={Distributed algorithms via gradient descent for Fisher markets},
  author={Birnbaum, Benjamin and Devanur, Nikhil R and Xiao, Lin},
  booktitle={Proceedings of the 12th ACM conference on Electronic commerce},
  pages={127--136},
  year={2011}
}

@inproceedings{wu2007proportional,
  title={Proportional response dynamics leads to market equilibrium},
  author={Wu, Fang and Zhang, Li},
  booktitle={Proceedings of the thirty-ninth annual ACM symposium on Theory of computing},
  pages={354--363},
  year={2007}
}

@inproceedings{cheung2025proportional,
  title={Proportional Response Dynamics in Gross Substitutes Markets},
  author={Cheung, Yun Kuen and Cole, Richard and Tao, Yixin},
  booktitle={Proceedings of the 26th ACM Conference on Economics and Computation},
  pages={389--389},
  year={2025}
}

@inproceedings{cutkosky2019anytime,
  title={Anytime online-to-batch, optimism and acceleration},
  author={Cutkosky, Ashok},
  booktitle={International conference on machine learning},
  pages={1446--1454},
  year={2019},
  organization={PMLR}
}

@article{zeng2022regularized,
  title={Regularized gradient descent ascent for two-player zero-sum Markov games},
  author={Zeng, Sihan and Doan, Thinh and Romberg, Justin},
  journal={Advances in Neural Information Processing Systems},
  volume={35},
  pages={34546--34558},
  year={2022}
}

@article{wu2025learning,
  title={Learning zero-sum linear quadratic games with improved sample complexity and last-iterate convergence},
  author={Wu, Jiduan and Barakat, Anas and Fatkhullin, Ilyas and He, Niao},
  journal={SIAM Journal on Control and Optimization},
  volume={63},
  number={5},
  pages={3244--3271},
  year={2025},
  publisher={SIAM}
}

@inproceedings{abe2023last,
  title={Last-Iterate Convergence with Full and Noisy Feedback in Two-Player Zero-Sum Games},
  author={Abe, Kenshi and Ariu, Kaito and Sakamoto, Mitsuki and Toyoshima, Kentaro and Iwasaki, Atsushi},
  booktitle={International Conference on Artificial Intelligence and Statistics},
  pages={7999--8028},
  year={2023},
  organization={PMLR}
}

@inproceedings{abe2022mutation,
  title={Mutation-driven follow the regularized leader for last-iterate convergence in zero-sum games},
  author={Abe, Kenshi and Sakamoto, Mitsuki and Iwasaki, Atsushi},
  booktitle={Uncertainty in Artificial Intelligence},
  pages={1--10},
  year={2022},
  organization={PMLR}
}

@article{dong2024uncoupled,
  title={Uncoupled and Convergent Learning in Monotone Games under Bandit Feedback},
  author={Dong, Jing and Wang, Baoxiang and Yu, Yaoliang},
  journal={arXiv preprint arXiv:2408.08395},
  year={2024}
}

@inproceedings{liupower,
  title={The Power of Regularization in Solving Extensive-Form Games},
  author={Liu, Mingyang and Ozdaglar, Asuman E and Yu, Tiancheng and Zhang, Kaiqing},
  booktitle={The Eleventh International Conference on Learning Representations},
  year={2023}
}

@article{park2023multi,
  title={Multi-player zero-sum markov games with networked separable interactions},
  author={Park, Chanwoo and Zhang, Kaiqing and Ozdaglar, Asuman},
  journal={Advances in Neural Information Processing Systems},
  volume={36},
  pages={37354--37369},
  year={2023}
}

@inproceedings{ozdaglar2021independent,
  title={Independent learning in stochastic games},
  author={Ozdaglar, Asuman and Sayin, Muhammed O and Zhang, Kaiqing},
  booktitle={International Congress of Mathematicians},
  year={2021}
}

@inproceedings{
abe2025boosting,
title={Boosting Perturbed Gradient Ascent for Last-Iterate Convergence in Games},
author={Kenshi Abe and Mitsuki Sakamoto and Kaito Ariu and Atsushi Iwasaki},
booktitle={The Thirteenth International Conference on Learning Representations},
year={2025},
url={https://openreview.net/forum?id=Jrt9iWalFy}
}

@inproceedings{abe2024adaptively,
  title={Adaptively Perturbed Mirror Descent for Learning in Games},
  author={Abe, Kenshi and Ariu, Kaito and Sakamoto, Mitsuki and Iwasaki, Atsushi},
  booktitle={International Conference on Machine Learning},
  pages={31--80},
  year={2024},
  organization={PMLR}
}

@article{giannou2021rate,
  title={On the rate of convergence of regularized learning in games: From bandits and uncertainty to optimism and beyond},
  author={Giannou, Angeliki and Vlatakis-Gkaragkounis, Emmanouil-Vasileios and Mertikopoulos, Panayotis},
  journal={Advances in Neural Information Processing Systems},
  volume={34},
  pages={22655--22666},
  year={2021}
}

@article{hsieh2022no,
  title={No-regret learning in games with noisy feedback: Faster rates and adaptivity via learning rate separation},
  author={Hsieh, Yu-Guan and Antonakopoulos, Kimon and Cevher, Volkan and Mertikopoulos, Panayotis},
  journal={Advances in Neural Information Processing Systems},
  volume={35},
  pages={6544--6556},
  year={2022}
}

@article{chen2023finite,
  title={A finite-sample analysis of payoff-based independent learning in zero-sum stochastic games},
  author={Chen, Zaiwei and Zhang, Kaiqing and Mazumdar, Eric and Ozdaglar, Asuman and Wierman, Adam},
  journal={Advances in Neural Information Processing Systems},
  volume={36},
  pages={75826--75883},
  year={2023}
}

@article{chen2024last,
  title={Last-Iterate Convergence of Payoff-Based Independent Learning in Zero-Sum Stochastic Games},
  author={Chen, Zaiwei and Zhang, Kaiqing and Mazumdar, Eric and Ozdaglar, Asuman and Wierman, Adam},
  journal={arXiv preprint arXiv:2409.01447},
  year={2024}
}

@inproceedings{cai2024accelerated,
  title={Accelerated algorithms for constrained nonconvex-nonconcave min-max optimization and comonotone inclusion},
  author={Cai, Yang and Oikonomou, Argyris and Zheng, Weiqiang},
  booktitle={Proceedings of the 41st International Conference on Machine Learning},
  pages={5312--5347},
  year={2024}
}

@article{jordan2024adaptive,
  title={Adaptive, doubly optimal no-regret learning in strongly monotone and exp-concave games with gradient feedback},
  author={Jordan, Michael and Lin, Tianyi and Zhou, Zhengyuan},
  journal={Operations Research},
  year={2024},
  publisher={INFORMS}
}

@article{ba2025doubly,
  title={Doubly optimal no-regret online learning in strongly monotone games with bandit feedback},
  author={Ba, Wenjia and Lin, Tianyi and Zhang, Jiawei and Zhou, Zhengyuan},
  journal={Operations Research},
  year={2025},
  publisher={INFORMS}
}

@article{cai2023uncoupled,
  title={Uncoupled and convergent learning in two-player zero-sum markov games with bandit feedback},
  author={Cai, Yang and Luo, Haipeng and Wei, Chen-Yu and Zheng, Weiqiang},
  journal={Advances in Neural Information Processing Systems},
  volume={36},
  pages={36364--36406},
  year={2023}
}

@article{freund1999adaptive,
  title={Adaptive game playing using multiplicative weights},
  author={Freund, Yoav and Schapire, Robert E},
  journal={Games and Economic Behavior},
  volume={29},
  number={1-2},
  pages={79--103},
  year={1999},
  publisher={Elsevier}
}

@inproceedings{bailey2018multiplicative,
  title={Multiplicative weights update in zero-sum games},
  author={Bailey, James P and Piliouras, Georgios},
  booktitle={Proceedings of the 2018 ACM Conference on Economics and Computation},
  pages={321--338},
  year={2018}
}

@inproceedings{liu2024comal,
  title={COMAL: A Convergent Meta-Algorithm for Aligning LLMs with General Preferences},
  author={Liu, Yixin and Oikonomou, Argyris and Zheng, Weiqiang and Cai, Yang and Cohan, Arman},
  booktitle={NeurIPS 2024 Workshop on Fine-Tuning in Modern Machine Learning: Principles and Scalability},
  year={2024}
}

@article{liu2025statistical,
  title={Statistical impossibility and possibility of aligning llms with human preferences: From condorcet paradox to nash equilibrium},
  author={Liu, Kaizhao and Long, Qi and Shi, Zhekun and Su, Weijie J and Xiao, Jiancong},
  journal={arXiv preprint arXiv:2503.10990},
  year={2025}
}

@inproceedings{swamy2024minimaximalist,
  title={A Minimaximalist Approach to Reinforcement Learning from Human Feedback},
  author={Swamy, Gokul and Dann, Christoph and Kidambi, Rahul and Wu, Steven and Agarwal, Alekh},
  booktitle={Forty-first International Conference on Machine Learning},
  year={2024}
}

@article{wu2024self,
  title={Self-play preference optimization for language model alignment},
  author={Wu, Yue and Sun, Zhiqing and Yuan, Huizhuo and Ji, Kaixuan and Yang, Yiming and Gu, Quanquan},
  journal={arXiv preprint arXiv:2405.00675},
  year={2024}
}

@article{ye2024theoretical,
  title={A theoretical analysis of nash learning from human feedback under general kl-regularized preference},
  author={Ye, Chenlu and Xiong, Wei and Zhang, Yuheng and Jiang, Nan and Zhang, Tong},
  journal={arXiv preprint arXiv:2402.07314},
  year={2024}
}

@article{meta2022human,
  title={Human-level play in the game of Diplomacy by combining language models with strategic reasoning},
  author={Meta Fundamental AI Research Diplomacy Team (FAIR)† and Bakhtin, Anton and Brown, Noam and Dinan, Emily and Farina, Gabriele and Flaherty, Colin and Fried, Daniel and Goff, Andrew and Gray, Jonathan and Hu, Hengyuan and others},
  journal={Science},
  volume={378},
  number={6624},
  pages={1067--1074},
  year={2022},
  publisher={American Association for the Advancement of Science}
}

@article{bowling2015heads,
  title={Heads-up limit hold’em poker is solved},
  author={Bowling, Michael and Burch, Neil and Johanson, Michael and Tammelin, Oskari},
  journal={Science},
  volume={347},
  number={6218},
  pages={145--149},
  year={2015},
  publisher={American Association for the Advancement of Science}
}

@article{brown2018superhuman,
  title={Superhuman AI for heads-up no-limit poker: Libratus beats top professionals},
  author={Brown, Noam and Sandholm, Tuomas},
  journal={Science},
  volume={359},
  number={6374},
  pages={418--424},
  year={2018},
  publisher={American Association for the Advancement of Science}
}

@article{brown2019superhuman,
  title={Superhuman AI for multiplayer poker},
  author={Brown, Noam and Sandholm, Tuomas},
  journal={Science},
  volume={365},
  number={6456},
  pages={885--890},
  year={2019},
  publisher={American Association for the Advancement of Science}
}

@inproceedings{mertikopoulos2018cycles,
  title={Cycles in adversarial regularized learning},
  author={Mertikopoulos, Panayotis and Papadimitriou, Christos and Piliouras, Georgios},
  booktitle={Proceedings of the twenty-ninth annual ACM-SIAM symposium on discrete algorithms},
  pages={2703--2717},
  year={2018},
  organization={SIAM}
}

@inproceedings{mao2024widetilde,
  title={ $\widetilde{O}(T^{-1})$ Convergence to (coarse) correlated equilibria in full-information general-sum Markov games},
  author={Mao, Weichao and Qiu, Haoran and Wang, Chen and Franke, Hubertus and Kalbarczyk, Zbigniew and Ba{\c{s}}ar, Tamer},
  booktitle={6th Annual Learning for Dynamics \& Control Conference},
  pages={361--374},
  year={2024},
  organization={PMLR}
}

@inproceedings{cai2024near,
  title={Near-optimal policy optimization for correlated equilibrium in general-sum markov games},
  author={Cai, Yang and Luo, Haipeng and Wei, Chen-Yu and Zheng, Weiqiang},
  booktitle={International Conference on Artificial Intelligence and Statistics},
  pages={3889--3897},
  year={2024},
  organization={PMLR}
}

@article{cen2021fast,
  title={Fast policy extragradient methods for competitive games with entropy regularization},
  author={Cen, Shicong and Wei, Yuting and Chi, Yuejie},
  journal={Advances in Neural Information Processing Systems},
  volume={34},
  pages={27952--27964},
  year={2021}
}

@article{cen2024fast,
  title={Fast policy extragradient methods for competitive games with entropy regularization},
  author={Cen, Shicong and Wei, Yuting and Chi, Yuejie},
  journal={Journal of machine learning Research},
  volume={25},
  number={4},
  pages={1--48},
  year={2024}
}

@inproceedings{cai2023doubly,
  title={Doubly optimal no-regret learning in monotone games},
  author={Cai, Yang and Zheng, Weiqiang},
  booktitle={International Conference on Machine Learning},
  pages={3507--3524},
  year={2023},
  organization={PMLR}
}

@inproceedings{cai2024fast,
title={Fast Last-Iterate Convergence of Learning in Games Requires Forgetful Algorithms},
author={Yang Cai and Gabriele Farina and Julien Grand-Cl{\'e}ment and Christian Kroer and Chung-Wei Lee and Haipeng Luo and Weiqiang Zheng},
booktitle={The Thirty-eighth Annual Conference on Neural Information Processing Systems},
year={2024},
url={https://openreview.net/forum?id=hK7XTpCtBi}
}

@article{howson1972equilibria,
  title={Equilibria of polymatrix games},
  author={Howson Jr, Joseph T},
  journal={Management Science},
  volume={18},
  number={5-part-1},
  pages={312--318},
  year={1972},
  publisher={INFORMS}
}

@article{janovskaja1968equilibrium,
  title={Equilibrium points in polymatrix games},
  author={Janovskaja, Elena},
  journal={Lithuanian Mathematical Journal},
  volume={8},
  number={2},
  pages={381--384},
  year={1968}
}

@article{cai2016zero,
  title={Zero-sum polymatrix games: A generalization of minmax},
  author={Cai, Yang and Candogan, Ozan and Daskalakis, Constantinos and Papadimitriou, Christos},
  journal={Mathematics of Operations Research},
  volume={41},
  number={2},
  pages={648--655},
  year={2016},
  publisher={INFORMS}
}

@inproceedings{soleymani2025faster,
  title={Faster Rates for No-Regret Learning in General Games via Cautious Optimism},
  author={Soleymani, Ashkan and Piliouras, Georgios and Farina, Gabriele},
  booktitle={Proceedings of the 56th Annual ACM Symposium on Theory of Computing},
  year={2025}
}

@article{arora2012multiplicative,
  title={The multiplicative weights update method: a meta-algorithm and applications},
  author={Arora, Sanjeev and Hazan, Elad and Kale, Satyen},
  journal={Theory of computing},
  volume={8},
  number={1},
  pages={121--164},
  year={2012},
  publisher={Theory of Computing Exchange}
}

@inproceedings{daskalakis2018limit,
  title={The limit points of (optimistic) gradient descent in min-max optimization},
  author={Daskalakis, Constantinos and Panageas, Ioannis},
  booktitle={the 32nd Annual Conference on Neural Information Processing Systems (NeurIPS)},
  year={2018}
}

@article{nash1950equilibrium,
  title={Equilibrium points in n-person games},
  author={Nash Jr, John F},
  journal={Proceedings of the national academy of sciences},
  volume={36},
  number={1},
  pages={48--49},
  year={1950},
  publisher={National Acad Sciences}
}

@article{perolat2022mastering,
  title={Mastering the game of Stratego with model-free multiagent reinforcement learning},
  author={Perolat, Julien and De Vylder, Bart and Hennes, Daniel and Tarassov, Eugene and Strub, Florian and de Boer, Vincent and Muller, Paul and Connor, Jerome T and Burch, Neil and Anthony, Thomas and others},
  journal={Science},
  volume={378},
  number={6623},
  pages={990--996},
  year={2022},
}

@article{silver2017mastering,
  title={Mastering the game of go without human knowledge},
  author={Silver, David and Schrittwieser, Julian and Simonyan, Karen and Antonoglou, Ioannis and Huang, Aja and Guez, Arthur and Hubert, Thomas and Baker, Lucas and Lai, Matthew and Bolton, Adrian and others},
  journal={nature},
  volume={550},
  number={7676},
  pages={354--359},
  year={2017},
  publisher={Nature Publishing Group}
}

@inproceedings{diakonikolas2020halpern,
  author    = {Jelena Diakonikolas},
  title     = {Halpern Iteration for Near-Optimal and Parameter-Free Monotone Inclusion and Strong Solutions to Variational Inequalities},
  booktitle = {Conference on Learning Theory (COLT)},
  year      = {2020},
}

@article{halpern1967fixed,
  title={Fixed points of nonexpanding maps},
  author={Halpern, Benjamin},
  journal={Bulletin of the American Mathematical Society},
  volume={73},
  number={6},
  pages={957--961},
  year={1967},
  publisher={American Mathematical Society}
}

@article{ouyang2021lower,
  title={Lower complexity bounds of first-order methods for convex-concave bilinear saddle-point problems},
  author={Ouyang, Yuyuan and Xu, Yangyang},
  journal={Mathematical Programming},
  volume={185},
  number={1},
  pages={1--35},
  year={2021},
  publisher={Springer}
}

@inproceedings{yoon2021accelerated,
  title={Accelerated Algorithms for Smooth Convex-Concave Minimax Problems with  $\mathcal{O}(1/k^2)$ Rate on Squared Gradient Norm},
  author={Yoon, TaeHo and Ryu, Ernest K},
  booktitle={International Conference on Machine Learning (ICML)},
  pages={12098--12109},
  year={2021},
  organization={PMLR}
}

@inproceedings{zinkevich2003online,
  title={Online convex programming and generalized infinitesimal gradient ascent},
  author={Zinkevich, Martin},
  booktitle={Proceedings of the 20th international conference on machine learning (ICML)},
  year={2003}
}

@inproceedings{cai2022finite,
  title={Finite-Time Last-Iterate Convergence for Learning in Multi-Player Games},
  author={Cai, Yang and Oikonomou, Argyris and Zheng, Weiqiang},
  booktitle={Advances in Neural Information Processing Systems (NeurIPS)},
  year={2022}
}

@inproceedings{hsieh2021adaptive,
  title={Adaptive learning in continuous games: Optimal regret bounds and convergence to Nash equilibrium},
  author={Hsieh, Yu-Guan and Antonakopoulos, Kimon and Mertikopoulos, Panayotis},
  booktitle={Conference on Learning Theory},
  pages={2388--2422},
  year={2021},
  organization={PMLR}
}

@inproceedings{daskalakis2011near,
  title={Near-optimal no-regret algorithms for zero-sum games},
  author={Daskalakis, Constantinos and Deckelbaum, Alan and Kim, Anthony},
  booktitle={Proceedings of the twenty-second annual ACM-SIAM symposium on Discrete Algorithms},
  pages={235--254},
  year={2011},
  organization={SIAM}
}

@inproceedings{wei2021linear,
  title={Linear Last-iterate Convergence in Constrained Saddle-point Optimization},
  author={Wei, Chen-Yu and Lee, Chung-Wei and Zhang, Mengxiao and Luo, Haipeng},
  booktitle={International Conference on Learning Representations (ICLR)},
  year={2021}
}

@article{rakhlin2013optimization,
  title={Optimization, learning, and games with predictable sequences},
  author={Rakhlin, Sasha and Sridharan, Karthik},
  journal={Advances in Neural Information Processing Systems},
  year={2013}
}

@inproceedings{anagnostides2022uncoupled,
  title={Uncoupled Learning Dynamics with $O(\log T)$ Swap Regret in Multiplayer Games},
  author={Anagnostides, Ioannis and Farina, Gabriele and Kroer, Christian and Lee, Chung-Wei and Luo, Haipeng and Sandholm, Tuomas},
  booktitle={Advances in Neural Information Processing Systems (NeurIPS)},
  year={2022}
}

@inproceedings{anagnostides2022near-optimal,
  title={Near-optimal no-regret learning for correlated equilibria in multi-player general-sum games},
  author={Anagnostides, Ioannis and Daskalakis, Constantinos and Farina, Gabriele and Fishelson, Maxwell and Golowich, Noah and Sandholm, Tuomas},
  booktitle={Proceedings of the 54th Annual ACM SIGACT Symposium on Theory of Computing (STOC)},
  year={2022}
}

@inproceedings{cai2011minmax,
  title={On minmax theorems for multiplayer games},
  author={Cai, Yang and Daskalakis, Constantinos},
  booktitle={Proceedings of the twenty-second annual ACM-SIAM symposium on Discrete algorithms (SODA)},
  year={2011}
}

@article{chen2020hedging,
  title={Hedging in games: Faster convergence of external and swap regrets},
  author={Chen, Xi and Peng, Binghui},
  journal={Advances in Neural Information Processing Systems (NeurIPS)},
  volume={33},
  pages={18990--18999},
  year={2020}
}

@inproceedings{daskalakis2019last,
  title={Last-Iterate Convergence: Zero-Sum Games and Constrained Min-Max Optimization},
  author={Daskalakis, Constantinos and Panageas, Ioannis},
  booktitle={10th Innovations in Theoretical Computer Science Conference (ITCS)},
  year={2019},
}

@inproceedings{daskalakis2018training,
  title={Training GANs with Optimism.},
  author={Daskalakis, Constantinos and Ilyas, Andrew and Syrgkanis, Vasilis and Zeng, Haoyang},
  booktitle={International Conference on Learning Representations (ICLR)},
  year={2018}
}

@article{daskalakis2021near-optimal,
  title={Near-optimal no-regret learning in general games},
  author={Daskalakis, Constantinos and Fishelson, Maxwell and Golowich, Noah},
  journal={Advances in Neural Information Processing Systems (NeurIPS)},
  year={2021}
}

@article{golowich2020tight,
  title={Tight last-iterate convergence rates for no-regret learning in multi-player games},
  author={Golowich, Noah and Pattathil, Sarath and Daskalakis, Constantinos},
  journal={Advances in neural information processing systems (NeurIPS)},
  year={2020}
}

@inproceedings{golowich2020last,
  title={Last iterate is slower than averaged iterate in smooth convex-concave saddle point problems},
  author={Golowich, Noah and Pattathil, Sarath and Daskalakis, Constantinos and Ozdaglar, Asuman},
  booktitle={Conference on Learning Theory (COLT)},
  year={2020}
}

@inproceedings{gorbunov2022last,
  title={Last-Iterate Convergence of Optimistic Gradient Method for Monotone Variational Inequalities},
  author={Gorbunov, Eduard and Taylor, Adrien and Gidel, Gauthier},
  booktitle={Advances in Neural Information Processing Systems},
  year={2022}
}

@article{syrgkanis2015fast,
  title={Fast convergence of regularized learning in games},
  author={Syrgkanis, Vasilis and Agarwal, Alekh and Luo, Haipeng and Schapire, Robert E},
  journal={Advances in Neural Information Processing Systems (NeurIPS)},
  year={2015}
}

@misc{munos_nash_2023,
	title = {Nash {Learning} from {Human} {Feedback}},
	url = {http://arxiv.org/abs/2312.00886},
	doi = {10.48550/arXiv.2312.00886},
	urldate = {2023-12-06},
	publisher = {arXiv},
	author = {Munos, Rémi and Valko, Michal and Calandriello, Daniele and Azar, Mohammad Gheshlaghi and Rowland, Mark and Guo, Zhaohan Daniel and Tang, Yunhao and Geist, Matthieu and Mesnard, Thomas and Michi, Andrea and Selvi, Marco and Girgin, Sertan and Momchev, Nikola and Bachem, Olivier and Mankowitz, Daniel J. and Precup, Doina and Piot, Bilal},
	month = dec,
	year = {2023},
	note = {arXiv:2312.00886 [cs, stat]},
}

@inproceedings{cai_accelerated_2023,
	title = {Accelerated {Single}-{Call} {Methods} for {Constrained} {Min}-{Max} {Optimization}},
	copyright = {All rights reserved},
	url = {https://openreview.net/forum?id=HRwN7IQLUKA},
	language = {en},
	urldate = {2023-03-21},
	booktitle = {Proceedings of the 11th {International} {Conference} on {Learning} {Representations}},
	author = {Cai, Yang and Zheng, Weiqiang},
	month = feb,
	year = {2023},
}

@misc{tammelin_solving_2014,
	title = {Solving {Large} {Imperfect} {Information} {Games} {Using} {CFR}+},
	url = {http://arxiv.org/abs/1407.5042},
	doi = {10.48550/arXiv.1407.5042},
	urldate = {2023-05-29},
	publisher = {arXiv},
	author = {Tammelin, Oskari},
	month = jul,
	year = {2014},
	note = {arXiv:1407.5042 [cs]},
}

@article{korpelevich_extragradient_1976,
	title = {The extragradient method for finding saddle points and other problems},
	volume = {12},
	url = {https://ci.nii.ac.jp/naid/10017556617/},
	urldate = {2022-01-20},
	journal = {Matecon},
	author = {Korpelevich, G. M.},
	year = {1976},
	pages = {747--756},
}

@article{popov_modification_1980,
	title = {A modification of the {Arrow}-{Hurwicz} method for search of saddle points},
	volume = {28},
	number = {5},
	journal = {Mathematical notes of the Academy of Sciences of the USSR},
	author = {Popov, Leonid Denisovich},
	year = {1980},
	note = {Publisher: Springer},
	pages = {845--848},
}

\appendix
\tableofcontents

\section{Proof of \Cref{lemma:average-iterate}}\label{sec:proof average-iterate}
\begin{proof}
By the zero-sum property and the linearity of the utility, we have
    \begin{align*}
        &\max_{x\in \times \Delta^{d_i}}\sum_{i=1}^n u_i(x_i, \bx^T_{-i}) - u_i(\bx^T) \\
        &= \sum_{i=1}^n  \max_{x_i\in \Delta^{d_i}}u_i(x_i, \bx^T_{-i}) \tag{$\sum_{i=1}^n u_i(\bx^T) = 0$} \\
        & =  \frac{1}{T} \sum_{i=1}^n \max_{x_i\in \Delta^{d_i}} \sum_{t=1}^T u_i(x_i, x^t_{-i}) \tag{Linearity of $u_i$ and $\bx^T_{-i} = \frac{1}{T}\sum_{t=1}^T x^t_{-i}$} \\
        & =\frac{1}{T}  \sum_{i=1}^n \InParentheses{ \max_{x_i\in \Delta^{d_i}} \sum_{t=1}^T u_i(x_i, x^t_{-i}) - \sum_{t=1}^T u_i(x^t) } \tag{$\sum_{i=1}^n u_i(x^t) = 0$ for all $t$} \\
        &= \frac{1}{T}  \sum_{i=1}^n \reg_i(T).
    \end{align*}
    This completes the proof. 
\end{proof}

{
\section{\AL Reduction for Proportional Response Dynamics in Fisher Markets}\label{app:PRD fisher}
The Fisher market is a classic model of allocating divisible goods, a special case of the Arrow-Debreu market.  In a Fisher market, there are 
$m$ agents and 
$n$ divisible goods. Each agent 
$i \in [n]$ has a budget 
$B_i > 0$ and each good $j \in [m]$ has a unit supply. Each agent $i$ also has a utility function $u_i: \-R^m_{\ge 0} \rightarrow \-R$ such that $u_i(x_i)$ is her utility of receiving the bundle $x_i \in \-R^m_{\ge 0}$ where $x_{ij} \ge 0$ is the amount of good $j$. Here we do not assume $u_i$ is linear. 

A price vector $p \in \-R^m_{\ge 0}$ specifies a price $p_j > 0$ for each good $j$. A \emph{Competitive Equilibrium (CE)} of a Fisher market is a pair of price vector and personalized allocations $(p, \{x_i\})$ that satisfies the following properties: 
\begin{itemize}
    \item[1.] \textit{Budget Feasible:}  $\InAngles{p, x_i} \le B_i$ for each agent $i$. 
    \item[2.] \textit{Utility Maximizing:} $u_i(x_i) = \max_{y_i \in \-R^m_{\ge 0}: \InAngles{p, y_i} \le e_i} u_i(y_i)$ for each agent $i$.
    \item[3.] \textit{Market Clears:} $\sum_{i} x_{ij} \le 1$ for each good $j$ and $\sum_{i} x_{ij} = 1$ if $p_j > 0$.
\end{itemize}

The \emph{proportional response dynamics (PRD)} is a distributed, independent dynamics leading to a CE in Fisher market. PRD works as follows:
\begin{itemize}
    \item \emph{Budget Spending:} each agent $i$ decides the allocation of her budget $b^t_{i}$, where $b^t_{ij}$ is her spending on good $j$. It must hold that $\sum_{j} b^t_{ij} = e_i$.
    \item \emph{Goods allocation:} the amount of good $j$ allocated to agent $i$ is proportional to their spending: $x^t_{ij} = \frac{b^t_{ij}}{\sum_i b^t_{ij}}$. Moreover, the price of each good $j$ is defined as the sum of agents’ spending $p^t_j = \sum_{i'} b^t_{i'j}$.
    \item \emph{Update:} Based on the allocation $x^t_i$, each agent $i$ then updates their spending of budget $b^{t+1}_i$ in the next round using the following rule:
    \begin{align*}
        b_{ij}^t = B_i \frac{x^t_{ij} \nabla_j u_i(x^t_i)}{\sum_{j'} x^t_{ij'} \nabla_{j'} u_i(x^t_i)}.
    \end{align*}
\end{itemize}

Although it has been shown that the price vector $\{p^t\}$ converges to a Market equilibrium price vector for Fisher markets in many settings including linear utilities~\citep{wu2007proportional,birnbaum2011distributed}, PRD for Fisher markets with Gross Substitutes utilities has a $O(1/T)$
convergence rate only in terms of the average price $\{\frac{1}{t}\sum_{k=1}^t p^k\}$, but not the last iterate price $\{p^t\}$~\citep{cheung2025proportional}.

\paragraph{\AL-PRD for Last-Iterate Convergence} 
We note that in PRD, the only feedback needed for their update is her allocation $x^t_i$, where good $j$'s allocation $x^t_{ij} = \frac{b^t_{ij}}{\sum_{i'} b^t_{i'j}}$ is proportional according to every agent's spending. Here, the key observation is that the feedback in PRD, $x^t_i$,  has a nice proportional structure with the spending $b^t_i$ despite each player having nonlinear utilities. We can then apply our \AL reduction. We just let each agent spend their averaged budget allocation $\bar{b}^t_i = \frac{1}{t}\sum_{k=1}^t b^k_i$ and observe the induced allocation $\bar{x}^t_i$. Due to the proportional structure, each agent can recover the original iterate $\{b^t_i\}_{i \in [n]}$'s allocation and then update according to the original PRD~\citep{cheung2025proportional}. In this way, we get a modified \AL-PRD whose last-iterate price vectors is equivalent to the averaged price vectors in the original PRD. The \AL-PRD enjoys $O(1/T)$ last-iterate convergence rate in the price vectors.}

\section{Last-Iterate Convergence with Gradient Feedback}\label{sec:proof last-gradient}
\subsection{Proof of \Cref{thm:last-gradient}}
\begin{proof}
    Let $\{x^t\}$ be the iterates of \OMWU dynamics with $\eta$.  Recall that we denote $u^t_i := u_i(\cdot, x^t_{-i})$. By \Cref{thm:RVU}, we have
    \begin{align*}
        \sum_{i=1}^n \reg_i(T) &\le \frac{\sum_{i=1}^n \log d_i}{\eta} + \sum_{i=1}^n \sum_{t=1}^T \InParentheses{\eta \InNorms{u^t_i - u^{t-1}_i}^2_\infty - \frac{1}{4\eta} \InNorms{x^t_i -x^{t-1}_i}_1^2 } \\
        & \le \frac{\sum_{i=1}^n \log d_i}{\eta},
    \end{align*}
    where we use \Cref{lemma:utility to action} and $\eta \le \frac{1}{2(n-1)}$ in the last inequalty. Invoking \Cref{lemma:average-iterate}, the average-iterate has total gap bounded by $\tgap(\bx^T) \le \frac{\sum_{i=1}^n \log d_i}{\eta T}$ for all $T$. By the \AL reduction guarantee in \Cref{thm:reduction}, we conclude the proof.
\end{proof}

\subsection{Proof of \Cref{corollary:dynamic}}
\begin{proof}
    By \Cref{thm:last-gradient}, we know $\bx^t$ is an  $\frac{\sum_{i=1}^n \log d_i}{\eta t}$-approximate Nash equilibrium for any $t \ge 1$. Thus, the dynamic regret of any player $i \in [n]$ is bounded by 
    \begin{align*}
        \Dreg^T_i = \sum_{t=1}^T \InParentheses{ \max_{x_i \in \Delta^{d_i}} u_i(x_i, \bx^t_{-i}) - u_i(\bx^t)} \le \sum_{t=1}^T  \frac{\sum_{i=1}^n \log d_i}{\eta t} = O\InParentheses{ \frac{\sum_{i=1}^n \log d_i}{\eta} \log T}. 
    \end{align*}
    This completes the proof.
\end{proof}

\section{Last-Iterate Convergence with Bandit Feedback}\label{sec:proof bandit}
 
\subsection{Analysis of Regret with Estimation Error}
We first analyze the sequence $\{x^t\}$ generated in the subroutine (\Cref{line:OMWU update}) of \Cref{alg:bandit}, which is updated according to \OMWU with the estimated utility sequence $\{\hu^t_{\varepsilon}\}$. 

We define $\Delta^t_i := \hU^t_{i, \varepsilon} - \bu^t_{i, \varepsilon}$ as the error of estimation in \Cref{line:sample}. The error of estimation for $u^t_i$, denoted as $\delta^t_i:= \hu^t_{i,\varepsilon} - u^t_i$, can be bounded as follows.
\begin{proposition}\label{prop:delta}
    For any $t \ge 1$ and $i$, $\InNorms{\delta^t_i}_\infty \le \InNorms{t\Delta^t_i} + \InNorms{(t-1)\Delta^{t-1}_i} + 2\varepsilon_t$. This further implies $\InNorms{\delta^t_i}_\infty^2 \le 3\InNorms{t\Delta^t_i}^2 + 3\InNorms{(t-1)\Delta^{t-1}_i}^2 + 12\varepsilon_t^2$.
\end{proposition}
\begin{proof} 
Let us define $u_{i,\mathrm{unif}}:=u_i(\cdot, \+U_{-i})$ the utility vector of $i$ when other players use the uniform strategy. Since the utility $u_i(x_i, x_{-i})$ is linear in $x_{-i}$, we have
\begin{align*}
    \bu^t_{i, \varepsilon} = u_i(\cdot, x^t_{-i, \varepsilon}) = (1-\varepsilon_t) u^t_{-i} + \varepsilon_t u_{i,\mathrm{unif}}, \forall t.
\end{align*}
Using the above equality and the definition of $\Delta^t_i$ and $\delta^t_i$, as well as $\varepsilon_t = t^{-1}$, we have 
\begin{align}\label{eq:delta}
    &t \Delta^t_i - (t-1) \Delta^{t-1}_i \nonumber\\
    &= t \cdot \hU^t_{i,\varepsilon} - (t-1)\cdot \hU^{t-1}_{i, \varepsilon} - t \cdot \bu^t_{i, \varepsilon} + (t-1) \cdot \bu^{t-1}_{i,\varepsilon} \nonumber\\
    &= \hu^t_{i, \varepsilon} - t \cdot \InParentheses{(1-\varepsilon_t) \bu^t_i + \varepsilon_t u_{i,\mathrm{unif}} } + (t-1) \cdot \InParentheses{(1-\varepsilon_{t-1})\bu^{t-1}_i + \varepsilon_{t-1} u_{i,\mathrm{unif}} } \nonumber\\
    &= \hu^t_{i, \varepsilon} - (1-\varepsilon_t) \sum_{k=1}^t u^k_i + (1-\varepsilon_{t-1}) \sum_{k=1}^{t-1} u^k_i \nonumber\\
    &= \hu^t_{i, \varepsilon} - u^t_i + \varepsilon_t u^t_i + (\varepsilon_t - \varepsilon_{t-1}) \sum_{k=1}^{t-1}u^k_i \nonumber\\
    &= \delta^t_i + \varepsilon_t u^t_i + (\varepsilon_t - \varepsilon_{t-1}) \sum_{k=1}^{t-1}u^k_i.
\end{align}
Since $\InNorms{u^k_i}_\infty \le 1$ for all $k$, the above equality implies
\begin{align*}
    \InNorms{\delta^t_i}_\infty &\le \InNorms{t\Delta^t_i} + \InNorms{(t-1)\Delta^{t-1}_i} + \varepsilon_t + (t-1)(\varepsilon_t - \varepsilon_{t-1}) \\
    &= \InNorms{t\Delta^t_i} + \InNorms{(t-1)\Delta^{t-1}_i} + 2\varepsilon_t. 
\end{align*}
We then apply the basic inequality $(a+b+c)^2 \le 3(a^2+b^2+c^2)$ to bound $\InNorms{\delta^t_i}_\infty^2$.  This completes the proof.
\end{proof}

In the following, we show that the regret against the original utility sequence $\{u^t\}$ can be bounded by RVU terms with an additional summation of second-order error terms. We remark that a naive analysis would lead to a summation of first-order terms, which gives a suboptimal convergence rate. 
\begin{lemma}\label{lemma:regret with error}
    Let $\{x^t_i\}$ be the iterates of \OMWU with step size $\eta > 0$ and utilities $\{ \hu^t_i\}$, we have 
    \begin{align*}
        \max_{x_i \in \Delta^{d_i}} \sum_{t=1}^T \InAngles{u^t_i, x^t_i - x_i} &\le \frac{\log d_i}{\eta} + 4\eta \sum_{t=1}^T \InNorms{u^t_i - u^{t-1}_i}^2_\infty - \frac{1}{8\eta}\sum_{t=1}^T \InNorms{x^t_i - x^{t-1}_i}^2_1 \\
        & \quad + 2\InNorms{T\Delta^T_i}_\infty + 26\eta \sum_{t=1}^{T-1} \InNorms{t \Delta^t_i}^2_\infty + 2\sum_{t=1}^T\varepsilon_t + 16\pi^2 \eta.
    \end{align*}
\end{lemma}
\begin{proof}
\textbf{Analysis}  By RVU property of \OMWU (\Cref{thm:RVU}), we can bound the regret for the estimated utilities $\{\hu^t_{i,\varepsilon}\}$: $\forall x \in \Delta^{d_i}$,
\begin{align*}
    \sum_{t=1}^T \InAngles{\hu^t_{i,\varepsilon}, x_i - x^t_i} \le \frac{\log d_i}{\eta} + \eta \sum_{t=1}^T \InNorms{\hu^t_{i,\varepsilon} - \hu^{t-1}_{i,\varepsilon}}^2_\infty - \frac{1}{4\eta}\sum_{t=1}^T \InNorms{x^t_i - x^{t-1}_i}^2_1
\end{align*}
Then by definition of $\hu^t_{i,\varepsilon} = u^t_i + \delta^t_i$, we can bound the regret for the true utilities $\{u^t_i\}$ by
\begin{align}
    \sum_{t=1}^T \InAngles{u^t_i, x_i - x^t_i} &\le \frac{\log d_i}{\eta} + 4\eta \sum_{t=1}^T \InNorms{u^t_i - u^{t-1}_i}^2_\infty - \frac{1}{4\eta}\sum_{t=1}^T \InNorms{x^t_i - x^{t-1}_i}^2_1 \nonumber \\
     &\quad +\underbrace{\sum_{t=1}^T \InAngles{\delta^t_i, x^t_i -x_i}}_{\texttt{I}} + \underbrace{4\eta \sum_{t=1}^T \InParentheses{\InNorms{\delta^t_i}^2_\infty + \InNorms{\delta^{t-1}_i}^2_\infty}}_{\texttt{II}}.\label{RVU with error I and II}
\end{align}
We note that the first three terms are standard terms in the RVU bound~\citep{syrgkanis2015fast}. We focus on the error terms \texttt{I} and \texttt{II}. The term \texttt{II} is a summation of second-order terms in $\InNorms{\delta^t_i}_\infty^2$ and we can bound it using \Cref{prop:delta} and $\sum_{t=1}^\infty \varepsilon_t^2 =\sum_{t=1}^\infty \frac{1}{t^2} =\frac{\pi^2}{6}$: 
\begin{align}
    \texttt{II} = 4\eta \sum_{t=1}^T \InParentheses{\InNorms{\delta^t_i}^2_\infty + \InNorms{\delta^{t-1}_i}^2_\infty} \le 24\eta \sum_{t=1}^T \InNorms{t \Delta^t_i}^2_\infty + 96\eta \sum_{t=1}^T \varepsilon_t^2 \le 24\eta \sum_{t=1}^T \InNorms{t \Delta^t_i}^2_\infty + 16\pi^2 \eta. \label{eq: error II}
\end{align}

However, a naive analysis of term \texttt{I} leads to a summation of first-order terms: by Cauchy-Schwarz, we have $ \texttt{I} \le \sum_{t=1}^T\InNorms{\delta^t_i}_\infty \InNorms{x^t_i-x_i}_1 \le 2\sum_{t=1}^T\InNorms{\delta^t_i}_\infty \le O(\sum_{t=1}^T \InNorms{t \Delta^t_i}_\infty) + O(\sum_{t=1}^T \varepsilon_t)$. If we tune the parameters optimally, the existence of the first-order error term $O(\sum_{t=1}^T \InNorms{t \Delta^t_i}_\infty)$ would result in an $O(T^{-\frac{1}{7}})$ rate, worse than the claimed $O(T^{-\frac{1}{5}})$ rate. 

\textbf{Improved Analysis of Term \texttt{I}} In the following, we give an improved analysis that bounds term $\texttt{I}$ by second-order error terms. The key is to rewrite term $\texttt{I}$ and use the negative term $-\sum_{t=1}^T \InNorms{x^t_i - x^{t-1}_i}_1^2$. 

By \eqref{eq:delta}, we have $$\delta^t_i = \hu^t_{i, \varepsilon} - u^t_i = t \Delta^t_i - (t-1) \Delta^{t-1}_i- \varepsilon_t u^t_i - (\varepsilon_t - \varepsilon_{t-1})\sum_{k=1}^{t-1}u^k_i.$$
We note that $\InNorms{\varepsilon_t u^t_i - (\varepsilon_t - \varepsilon_{t-1})\sum_{k=1}^{t-1}u^k_i}_\infty \le 2\varepsilon_t$ as the utility is bounded in $[0,1]$. This implies
\begin{align*}
    \texttt{I} &= \sum_{t=1}^T \InAngles{\delta^t_i, x^t_i - x_i} \\
    &\le \sum_{t=1}^T \InAngles{t\Delta^t_i- (t-1)\Delta^{t-1}_i, x^t_i - x_i} + \sum_{t=1}^T \InNorms{\varepsilon_t u^t_i - (\varepsilon_t - \varepsilon_{t-1})\sum_{k=1}^{t-1}u^k_i}_\infty \cdot \InNorms{x^t_i - x_i}_1.   \\
    &\le \InAngles{T \Delta^T_i, x^T_i - x_i} + \sum_{t=1}^{T-1} \InAngles{t \Delta^t_i, x^t_i - x^{t+1}_i} + 4\sum_{t=1}^T\varepsilon_t \\
    &\le 2\InNorms{T\Delta^T_i}_\infty + \sum_{t=1}^{T-1} \InNorms{t \Delta^t_i}_\infty \InNorms{x^t_i - x^{t+1}_i}_1+ 4\sum_{t=1}^T\varepsilon_t.
\end{align*}
As a result, we can combine $-\frac{1}{8\eta}\sum_{t=1}^T\InNorms{x^t_i -x^{t-1}_i}_1^2$ and using basic inequality $\InAngles{a,b} - \InNorms{b}^2 \le \frac{1}{4} \InNorms{a}^2$ to get 
\begin{align}
    \texttt{I} -\frac{1}{8\eta}\sum_{t=1}^T\InNorms{x^t_i -x^{t-1}_i}_1^2 &\le 2\InNorms{T\Delta^T_i}_\infty + \sum_{t=1}^{T-1} \InNorms{t \Delta^t_i}_\infty \InNorms{x^t_i - x^{t+1}_i}_1 -\frac{1}{8\eta}\sum_{t=1}^T\InNorms{x^t_i -x^{t-1}_i}_1^2 +4\sum_{t=1}^T\varepsilon_t\nonumber \\
    &\le 2\InNorms{T\Delta^T_i}_\infty + 2\eta \sum_{t=1}^{T-1} \InNorms{t \Delta^t_i}^2_\infty+4\sum_{t=1}^T\varepsilon_t. \label{eq:error I}
\end{align}

Combining \eqref{RVU with error I and II}, \eqref{eq: error II}, and \eqref{eq:error I}, we get 
\begin{align*}
    \max_{x_i \in \Delta^{d_i}} \sum_{t=1}^T \InAngles{u^t_i, x^t_i - x_i} &\le \frac{\log d_i}{\eta} + 4\eta \sum_{t=1}^T \InNorms{u^t_i - u^{t-1}_i}^2_\infty - \frac{1}{8\eta}\sum_{t=1}^T \InNorms{x^t_i - x^{t-1}_i}^2_1 \\
    & \quad + 2\InNorms{T\Delta^T_i}_\infty + 26\eta \sum_{t=1}^{T-1} \InNorms{t \Delta^t_i}^2_\infty + 4\sum_{t=1}^T\varepsilon_t + 16\pi^2 \eta.
\end{align*}
This completes the proof.
\end{proof}
\subsection{Proof of \Cref{thm:last-bandit}}

Recall that $d = \max_{i\in[n]} d_i$. By \Cref{lemma:regret with error}, we can bound the social regret of all players with respect to the iterates $\{x^t\}$ and utilities $\{u^t\}$ by 
\begin{align}
    \sum_{i=1}^n \reg_i(T, \{x^t_i\}, \{u^t_i\}) &\le \frac{n\log  d}{\eta} + \underbrace{ \sum_{i=1}^n \sum_{t=1}^T \InParentheses{4\eta \InNorms{u^t_i - u^{t-1}_i}^2_\infty - \frac{1}{8\eta} \InNorms{x^t_i - x^{t-1}_i}^2_1}}_{\texttt{I}} \nonumber \\
    & \quad + \underbrace{\sum_{i=1}^n \InParentheses{ 2 \InNorms{T \Delta^T_i}_\infty + 26 \eta \sum_{t=1}^{T-1} \InNorms{t \Delta^t_i}^2_\infty}}_{\texttt{II}} + 4n\sum_{t=1}^T\varepsilon_t + 16n\pi^2 \eta.  \label{eq:bandit-1}
\end{align}
By \Cref{lemma:utility to action} and the fact that $\eta \le \frac{1}{6n}$, we have $\texttt{I} \le 0$. 

In the following, we analyze term \texttt{II} and conclude the last-iterate convergence rate of \Cref{alg:bandit} for $\{\varepsilon_t = t^{-1}\}$ and different choices of $\{B_t\}$. All these choices give an $O(k^{-\frac{1}{5}})$ last-iterate convergence rate. The difference is that the choices $\{B_t = dt^{4}\}$ lead to better dependence on $d$ or $\delta$ but require knowledge about $d$ or $\delta$, while the choice of $\{B_t  = t^{4}\}$ is fully uncoupled as each player does not need to know even an upper bound of $d$.

\paragraph{Case 1: $B_t = d\cdot t^{4}$}
We note that for $T = O(\log(\frac{1}{\delta}))$, it trivially holds that $\sum_{i=1}^n \reg_i(T, \{x^t_i\}, \{u^t_i\})\le nT = O(n\log(\frac{1}{\delta}))$. Thus $\bar{x}^T$ is an $O(\frac{n\log(\frac{1}{\delta})}{T})$-approximate Nash equilibrium for all $T = O(\log\frac{1}{\delta})$. In the following, we focus on $T \ge \log(\frac{1}{\delta})$.

By \Cref{lemma: all T estimation error}, we have with probability $1 - O(nd^3\delta)$, 
\begin{align*}
    \max_{i\in[n]} \InNorms{\Delta^t_i}_\infty \le  2\sqrt{\frac{d \log\InParentheses{\frac{B_t t^2}{\delta}}}{B_t \varepsilon_t}}, \forall t \ge \log\InParentheses{\frac{1}{\delta}}. 
\end{align*}
For $t \le d\log(\frac{1}{\delta})$, we can use the trivial bound of $\InNorms{\Delta^t_i}_\infty \le 1$. Then we can bound term \texttt{II} by
\begin{align}
    \texttt{II} \le 4nT\sqrt{\frac{d \log\InParentheses{\frac{B_T T^2}{\delta}}}{B_T \varepsilon_T}}+  104n\eta\sum_{t=1}^{T-1} \frac{t^2 \cdot d \log\InParentheses{\frac{B_t t^2}{\delta}}}{B_t \varepsilon_t} + 26\eta\log\InParentheses{\frac{1}{\delta}} .\label{eq:bandit-2}
\end{align}

Combining \eqref{eq:bandit-1} and \eqref{eq:bandit-2}, and $B_t = d\cdot t^{4}$ and $\varepsilon_t = t^{-1} $, we get 
\begin{align*}
    &\sum_{i=1}^n \reg_i(T, \{x^t_i\}, \{u^t_i\})\\ &\le \frac{n \log d}{\eta} + 4nT\sqrt{\frac{d \log\InParentheses{\frac{B_T T^2}{\delta}}}{B_T \varepsilon_T}}+ 104n\eta \sum_{t=1}^{T-1} \frac{t^2 \cdot d \log\InParentheses{\frac{B_t t^2}{\delta}}}{B_t \varepsilon_t}+  26\eta\log\InParentheses{\frac{1}{\delta}}  +4n\sum_{t=1}^T\varepsilon_t + 16n\pi^2 \eta \\
    & = O\InParentheses{\frac{n \log^2(\frac{dT}{\delta})}{\eta}}.
\end{align*}
 By \Cref{lemma:average-iterate}, we have the average iterate $\bx^T:= \frac{1}{T} \sum_{t=1}^T x^t$ is an $O(\frac{n\log^2(\frac{dT}{\delta})}{\eta T})$-approximate Nash equilibrium. Since $\varepsilon_T = \frac{1}{T}$, then we know $\bx^T_\varepsilon := (1-\varepsilon_T) \bx^T + \varepsilon_T \cdot \otimes_{i=1}^n\mathrm{Uniform}(d_i)$ is also an $O(\frac{n\log^2(\frac{dT}{\delta})}{\eta T})$-approximate Nash equilibrium for any $T \ge 1$. 
 
By the choice of the epoch length $B_t = d\cdot t^{4}$ , we have that the actual sequence $\{\by^k\}$ generated by \AL-\OMWU-\Bandit satisfies $\by^k_\varepsilon = \bx^{t_k}_\varepsilon$ where $t_k = \Theta (d^{-\frac{1}{5}}k^{\frac{1}{5}})$. Thus we can conclude that with probability at least $1 - O(nd^3\delta)$, it holds for all $k \ge 1$ that $\by^k$ is a $\beta_k$-approximate Nash equilibrium with 
\begin{align*}
    \beta_k = O\InParentheses{\frac{n\log^2(\frac{kd}{\delta})}{\eta} \cdot d^{\frac{1}{5}}k^{-\frac{1}{5}} }.
\end{align*}
This completes the proof.

\paragraph{Case 2: $B_t = t^4$} The proof is very similar to case 1.  We note that for $T = O(\log(\frac{1}{\delta}))$, it trivially holds that $\sum_{i=1}^n \reg_i(T, \{x^t_i\}, \{u^t_i\})\le nT = O(n\log(\frac{1}{\delta}))$. Thus $\bar{x}^T$ is an $O(\frac{n\log(\frac{1}{\delta})}{T})$-approximate Nash equilibrium for all $T = O(\log\frac{1}{\delta})$. In the following, we focus on $T \ge \log(\frac{1}{\delta})$. By item 2 in \Cref{lemma: all T estimation error}, with probability $1 - O(nd^3\delta)$, for all $t \ge d \log(\frac{1}{\delta})$, 
\begin{align*}
    \max_{i\in[n]} \InNorms{\Delta^t_i}_\infty \le  2\sqrt{\frac{d \log\InParentheses{\frac{B_t t^2}{\delta}}}{B_t \varepsilon_t}}.
\end{align*}
For $t \le \log(\frac{1}{\delta})$, we can use the trivial bound of $\InNorms{\Delta^t_i}_\infty \le 1$. These gives
\begin{align}
    \texttt{II} \le 4nT\sqrt{\frac{d \log\InParentheses{\frac{B_T T^2}{\delta}}}{B_T \varepsilon_T}}+  104n\eta\sum_{t=1}^{T-1} \frac{t^2 \cdot d \log\InParentheses{\frac{B_t t^2}{\delta}}}{B_t \varepsilon_t} + 26\eta \log (\frac{1}{\delta}) . \label{eq:bandit-3}
\end{align}

Now can combine \eqref{eq:bandit-1} and \eqref{eq:bandit-3} with $B_t = t^4$ and $\varepsilon_t = t^{-1}$ and get
\begin{align*}
    &\sum_{i=1}^n \reg_i(T, \{x^t_i\}, \{u^t_i\})\\
    &\le \frac{n \log d}{\eta} + 4nT\sqrt{\frac{d \log\InParentheses{\frac{B_T T^2}{\delta}}}{B_T \varepsilon_T}}+ 104n\eta \sum_{t=1}^{T-1} \frac{t^2 \cdot d \log\InParentheses{\frac{B_t t^2}{\delta}}}{B_t \varepsilon_t}+26\eta\log (\frac{1}{\delta}) +4n\sum_{t=1}^T\varepsilon_t + 16n\pi^2 \eta\\
    & = O\InParentheses{\frac{nd \log^2(\frac{dT}{\delta})}{\eta}}.
\end{align*}
Compared to case 1, this regret bound has an additional $d$ dependence.

Similar to the analysis in the former case, we have $\bar{x}^T_\varepsilon$ is an $O(\frac{nd\log^2(\frac{dT}{\delta})}{\eta})$-approximate Nash equilibrium for all $T \ge 1$. By the choice of $B_t = t^4$, we have that the actual sequence $\{\by^k\}$ generated by \AL-\OMWU-\Bandit satisfies $\by^k_\varepsilon = \bx^{t_k}_\varepsilon$ where $t_k = \Theta (k^{\frac{1}{5}})$. Thus we can conclude that with probability at least $1 - O(nd^3\delta)$, it holds for all $k \ge 1$ that $\by^k$ is a $\beta_k$-approximate Nash equilibrium with 
\begin{align*}
    \beta_k = O\InParentheses{\frac{n\log^2(\frac{dk}{\delta})}{\eta} \cdot dk^{-\frac{1}{5}} }.
\end{align*}
This completes the proof.

\subsection{Estimation}\label{sec:estimation}
In this subsection, we analyze the estimation error $\Delta^t_i := \hU^t_{i,\varepsilon} - \bu^t_{i,\varepsilon}$ for any $i \in [n]$. 
Recall that $d:= \max_{i \in [n]} d_i$. Within the epoch of length $B_t$, with high probability, each action receives at least $\Omega(\frac{B_t \varepsilon_t}{d})$ samples. We have with high probability that $|(\hU^t_{i,\varepsilon}- \bu^t_{i, \varepsilon})[a]| \le \Tilde{O}(\frac{\sqrt{d}}{\sqrt{B_t \varepsilon_t}})$ for all action $a$, which implies $\InNorms{\Delta^t_i}_\infty \le \Tilde{O}(\sqrt{\frac{d}{B_t \varepsilon_t}})$. The formal guarantees are as follows.

\begin{lemma}
\label{lemma:estimation error}
    In \Cref{alg:bandit}, with probability $1 - \frac{d\delta}{t^2} - d\exp\InParentheses{-\frac{B_t\varepsilon_t^2}{2d^2}}$, we have 
    \begin{align*}
        \InNorms{\Delta^t_{i, \varepsilon}}_\infty = \InNorms{\hU^t_i - \bu^t_i}_\infty \le 2\sqrt{\frac{d \log\InParentheses{\frac{B_t t^2}{\delta}}}{B_t \varepsilon_t}}.
    \end{align*}
\end{lemma}
\begin{proof}
    Define $N^t_a :=\sum_{k=1}^{B_t} \mathbb{I}[a^k  = a]$ to be the number of samples for action $a$. Since $\Pr[a^k = a] \ge \frac{\varepsilon_t}{d_i} \ge \frac{\varepsilon_t}{d}$, we have $\-E \InBrackets{N^t_a} \ge \frac{B_t \varepsilon_t}{d}$. By Hoeffding's inequality, we have 
    \begin{align*}
        \Pr\InBrackets{N^t_a \le \frac{B_t \varepsilon_t}{2d}} \le \exp\InParentheses{-\frac{B_t\varepsilon_t^2}{2d^2}}.
    \end{align*}
    Thus with probability $1 - d\exp\InParentheses{-\frac{B_t\varepsilon_t^2}{2d^2}}$, every action has been sampled at least $\frac{B_t \varepsilon_t}{2d}$ times, i.e., $N^t_a \ge \frac{B_t\varepsilon_t}{2d}$.

    Note that the number of samples $N^t_{a}$ for action $a$ is a random variable, so we can not directly use Azuma-Hoeffding's inequality to argue $| \hU_{i, \varepsilon}^t[a] - \bu_{i, \varepsilon}^t[a]| \le \sqrt{\frac{d\log(1/\delta)}{B_t\varepsilon_t}}$ with probability $1 - \delta$.  We define a sequence of random variables where $a^k_{-i} \sim \bx^t_{-i, \varepsilon}$ and $r^k = u_i(a, a^k_{-i})$, and define $\hU^t_m[a] := \frac{1}{m}\sum_{k=1}^m r^k$. Thus $\hU^t_m$ is an unbiased estimator for $\bu^t_{i, \varepsilon}[a]$ for all $m \in [1, B_t]$. Then we can use Azuma-Hoeffding's inequality to get that 
    \begin{align*}
       &\Pr\InBrackets{|\hU_{i, \varepsilon}^t[a] - \bu_{i, \varepsilon}^t[a]| \ge \sqrt{\frac{2\log\InParentheses{\frac{B_t t^2}{\delta}}}{N^t_a}} } \\
       &\le \Pr\InBrackets{\exists m \in [B_t], |\hU^t_m[a] - \bu^t_{i,\varepsilon}[a]| \ge \sqrt{\frac{2\log\InParentheses{\frac{B_t t^2}{\delta}}}{m}}  } \\
        & \le \sum_{m=1}^{B_t} \Pr\InBrackets{|\hU^t_m[a] - \bu^t_{i,\varepsilon}[a]| \ge \sqrt{\frac{2\log\InParentheses{\frac{B_t t^2}{\delta}}}{m}}  }\\
        &\le \sum_{m=1}^{B_t} \frac{\delta}{B_t t^2} = \frac{\delta}{t^2}.
    \end{align*}
    
    Recall that with probability $1 - d\exp\InParentheses{-\frac{B_t\varepsilon_t^2}{2d^2}}$, $N^t_a \ge \frac{B_t\varepsilon_t}{2d}$ for all $a$. By a union bound over actions, we get with probability $1 - \frac{d\delta}{t^2} - d\exp\InParentheses{-\frac{B_t\varepsilon_t^2}{2d^2}}$, 
    \begin{align*}
        \InNorms{\hU_{i, \varepsilon}^t - \bu_{i,\varepsilon}^t}_\infty \le 2\sqrt{\frac{d \log\InParentheses{\frac{B_t t^2}{\delta}}}{B_t \varepsilon_t}}.
    \end{align*}
    This completes the proof.
\end{proof}

Using a union bound over all players $i \in [n]$ and epoch $t \ge 1$, we have the following guarantee.
\begin{lemma}\label{lemma: all T estimation error}
Consider a polymatrix game in which each player has at most $d$ actions. Let $\{\Delta_i^t=U^t_{i,\varepsilon}-u^t_{i,\varepsilon}\}_{t\ge 1}$ denote the estimation error vector produced by \Cref{alg:bandit} for player $i$ at round $t$.
Suppose every player runs \Cref{alg:bandit} with $B_t \ge t^{4}$ and $\varepsilon_t = t^{-1}$. Then, with probability at least $1 - O\InParentheses{nd^3\delta}$, simultaneously for all players $i\in[n]$ and all rounds $t\ge \log(\frac{1}{\delta})$,
\begin{align*}
    \InNorms{\Delta_i^t}_\infty \le 2\sqrt{\frac{d\log\InParentheses{\frac{B_tt^2}{\delta}}}{B_t\varepsilon_t}} .
\end{align*}
\end{lemma}
\begin{proof}
By \Cref{lemma:estimation error}, for each fixed player $i\in[n]$ and round $t\ge 1$, we have
\begin{align*}
    \Pr\InParentheses{\InNorms{\Delta_i^t}_\infty > 2\sqrt{\frac{d\log\InParentheses{\frac{B_tt^2}{\delta}}}{B_t\varepsilon_t}}} \le \frac{d\delta}{t^2}+
     d\exp\InParentheses{-\frac{B_t\varepsilon_t^2}{2d^2}}.
\end{align*}
Applying a union bound over all players $i\in[n]$ and all rounds $t\ge \log(\frac{1}{\delta})$ gives that the claim holds with probability at least
\begin{align*}
    1
    - nd\delta \sum_{t=1}^{\infty}\frac{1}{t^2}
    - nd \sum_{t=\log(\frac{1}{\delta})}^{\infty}\exp\InParentheses{-\frac{B_t\varepsilon_t^2}{2d^2}}.
\end{align*}
We have $\sum_{t=1}^\infty \frac{1}{t^2} = O(1)$. With $\varepsilon_t = t^{-1}$ and $B_t \ge t^{4}$, we have $B_t\varepsilon_t^2 \ge t^2$, so the second summation is bounded by
\[
    nd\sum_{t=\log(\frac{1}{\delta})}^{\infty}\exp\InParentheses{-\frac{t^2}{2d^2}}
    \le nd\delta \sum_{t= 1}^{\infty}\exp\InParentheses{-\frac{t}{2d^2}} \le O(nd^3 \delta).
\]
So the overall failure probability is at most $O\InParentheses{nd^3\delta}$. This completes the proof.
\end{proof}

\section{Robustness against Adversaries in the Bandit Setting}\label{app:bandit robustness}
Consider the following utility estimation subroutine equipped with \Cref{alg:bandit}. Here, we no longer assume the other players employ the same algorithm, and they might behave adversarially. To this end, in each epoch $t$, we denote the true utility vectors for player $i$ as $\{u_i^{t,j} \in [0,1]^{d_i}\}_{j \in [B_t]}$ and the true accumulated utility within epoch $t$ as
\[
U^t_i = \sum_{j=1}^{B_t} u^{t,j}_i
\]
In each epoch $t$, the player plays actions $\{a^{t,j} \sim x^t_{i, \varepsilon}\}$ for $j \in [B_t]$ and receives $\{r^{t,j}=u_i^{t,j}[a^j]\}_{j \in [B_t]}$, it also maintains an importance weighted estimator of the accumulated utility within epoch $B_t$:
\begin{align*}
    \tu^{t,j}_i[b] = \frac{r^j\cdot \boldsymbol{1}[b = a^j]}{x^t_{i,\varepsilon}[b]}, \forall b \in [d_i], \quad \tU^t_{i} &= \sum_{j=1}^{B_t}  \tu^{t,j}_i.
\end{align*}
Denote $T_t = \sum_{k=1}^t B_t$, the player maintains estimated regret 
\begin{align*}
     \widetilde{\reg}^{T_t} = \max_{x\in \Delta^{d_i}}\sum_{k=1}^t \InAngles{\tU^k_i, x} - \sum_{k=1}^t \InAngles{\tU^k_i, x^k_{i,\varepsilon}}, 
\end{align*}
while the true regret is 
\begin{align*}
    \reg^T = \max_{x\in \Delta^{d_i}}\sum_{k=1}^t \InAngles{U^k_i, x} -  \sum_{k=1}^t \InAngles{U^k_i, x^t_{i,\varepsilon}}. 
\end{align*}
Then we have the following
\begin{proposition}
    Fix any $\delta >0$. With probability at least $1 -\delta$, for any $t \ge 1$ and $T_t = \sum_{k=1}^t B_t$, we have
    \begin{align*}
     \reg^{T_t} \in \InBrackets{\widetilde{\reg}^{T_t} - 4d_i\sqrt{\sum_{k=1}^t (k^2 B_k)}\log\InParentheses{\frac{\pi^2d_i t^2}{3\delta}}, \widetilde{\reg}^{T_t} + 4d_i\sqrt{\sum_{k=1}^t (k^2 B_k)}\log\InParentheses{\frac{\pi^2d_i t^2}{3\delta}}} 
\end{align*}
\end{proposition}
\begin{proof}
The importance weighed estimator is unbiased as $\-E[\tu^{t,j}_i[b]] = u^{t, j}_i[b]$ for all $b \in [d_i]$. Moreover, since $x^t_{i,\varepsilon}[b] \ge \frac{1}{d_i t}$, we have $
    -1 \le \tu^{t,j}_i[b] - u^{t,j}_i[b] \le d_i t$.
By Azuma-Hoeffdind inequality, we have for any $t \ge 1$ and $b \in [d_i]$, 
\begin{align*}
    \Pr\InBrackets{\left |\sum_{k=1}^t(\tU^k_i[b] - U^k_i[b])\right | \le 2d_i\sqrt{\sum_{k=1}^t (k^2 B_k)}\log\InParentheses{\frac{\pi^2d_i t^2}{3\delta}}} \ge 1- \frac{6\delta}{\pi^2 d_i t^2}.
\end{align*}
Using a union bound over $t \ge 1$ and $b \in [d_i]$ gives
\begin{align*}
    \Pr\InBrackets{ \forall t \ge 1, \InNorms{\sum_{k=1}^t(\tU^k_i - U^k_i)}_\infty \le 2d_i\sqrt{\sum_{k=1}^t (k^2 B_k)}\log\InParentheses{\frac{\pi^2d_i t^2}{3\delta}}} \ge 1 -\delta.
\end{align*}

Since the $|\max v - \max v'| \le \InNorms{v - v}_\infty$, we can bound the error in regret estimation as 
\begin{align*}
  \widetilde{\reg}^T - 2 \InNorms{\sum_{k=1}^t(\tU^k_i - U^k_i)}_\infty \le  \reg^T \le \widetilde{\reg}^T + 2 \InNorms{\sum_{k=1}^t(\tU^k_i - U^k_i)}_\infty.
\end{align*}
This implies with probability at least $1 - \delta$, for all $t \ge 1$ and $T_t = \sum_{k=1}^t B_k$, we have
\begin{align*}
     \reg^{T_t} \in \InBrackets{\widetilde{\reg}^T - 4d_i\sqrt{\sum_{k=1}^t (k^2 B_k)}\log\InParentheses{\frac{\pi^2d_i t^2}{3\delta}}, \widetilde{\reg}^T + 4d_i\sqrt{\sum_{k=1}^t (k^2 B_k)}\log\InParentheses{\frac{\pi^2d_i t^2}{3\delta}}}.
\end{align*}
This completes the proof.
\end{proof}

Choosing $B_t = t^4$, we have $T_t = \sum_{k=1}^t B_k = \theta(t^5)$ and $t = \Theta(T_t^{\frac{1}{5}})$. We have, with probability $1- \delta$, for all $T$,
\begin{align*}
     \reg^{T_t} = \widetilde{\reg}^{T_t} \pm \Theta\InParentheses{d_it^{3.5} \log\InParentheses{\frac{d_i t}{\delta}}} = \widetilde{\reg}^{T_t} \pm \Theta\InParentheses{d_i T_t^{\frac{7}{10}} \log\InParentheses{\frac{d_i T_t}{\delta}}}. 
\end{align*}
We note that for general $T \in [T_t, T_{t+1}]$ (denote $T_0 =1$), we have $B_{t+1} = (t+1)^4 = \Theta(T_t^{\frac{4}{5}})$
\begin{align*}
    \reg^T &:=\max_{x \in \Delta^{d_i}} \InParentheses{\sum_{k=1}^t \InAngles{U^k_i, x} + \sum_{j=1}^{T - T_t} \InAngles{u^{t+1,j}_i, x}} - \InParentheses{ \sum_{k=1}^t \InAngles{U^k_i, x^k_{i,\varepsilon}} + \sum_{j=1}^{T - T_t} \InAngles{u^{t+1,j}_i, x^{t+1}_{i,\varepsilon}} }
    \\
    & \le \reg^{T_t} + B_{t+1} \\
    &= \widetilde{\reg}^{T_t} +\max\left\{\Theta\InParentheses{T_t^{\frac{4}{5}}}, \Theta\InParentheses{d_i T_t^{\frac{7}{10}} \log\InParentheses{\frac{d_i T_t}{\delta}}} \right \} \\
    &= \widetilde{\reg}^{T_t} + \max\left\{\Theta\InParentheses{T^{\frac{4}{5}}}, \Theta\InParentheses{d_i T^{\frac{7}{10}} \log\InParentheses{\frac{d_i T}{\delta}}} \right \}.
\end{align*}
Thus whenever $\widetilde{\reg}^{T_t} = O(T_t^{\frac{4}{5}})$, we have $\reg^{T} = O(T^{\frac{4}{5}})$ for all iterations $T \in [T_t, T_{t+1}]$.

\paragraph{Robustness in Adversarial Setting} The player could run \Cref{alg:bandit} and track its regret using the above estimation $\widetilde{\reg}^{T_t}$ at the end of each epoch $t$. Whenever she detects that $\widetilde{\reg}^{T_t} = \Omega(T_t^{\frac{4}{5}})$, which cannot happen if all the players employ \Cref{alg:bandit}, she can switch to a standard no-regret bandit algorithm with worst-case optimal regret. This procedure guarantees a worst-case $\widetilde{O}(T^{\frac{4}{5}})$ regret.

\end{document}